%% file: main.tex
\documentclass[11pt]{article}

\usepackage{amsmath,amssymb,amsthm}
\usepackage{mathtools}
\usepackage{booktabs}
\usepackage{microtype}
\usepackage[round]{natbib}
\usepackage[letterpaper,textwidth=360pt,textheight=600pt,centering,%
            headheight=14pt,headsep=18pt,footskip=24pt]{geometry}
\usepackage[hidelinks]{hyperref}
\hypersetup{
  pdftitle={Forced Coincidence and Its Boundary in Finite Condition Structures},
  pdfauthor={Daisuke Hirota},
  pdfsubject={Finite condition structures; forced coincidence; separators; load},
  bookmarksnumbered=true,
}
\usepackage{fancyhdr}
\input{numbers.tex}

\newtheorem{theorem}{Theorem}[section]
\newtheorem{proposition}[theorem]{Proposition}
\newtheorem{lemma}[theorem]{Lemma}
\newtheorem{corollary}[theorem]{Corollary}

\theoremstyle{definition}
\newtheorem{definition}[theorem]{Definition}

\theoremstyle{remark}
\newtheorem{remark}[theorem]{Remark}

\title{Forced Coincidence and Its Boundary\\
in Finite Condition Structures}
\author{Daisuke Hirota}
\date{}

\begin{document}
\maketitle

\begin{abstract}
Accounts of organizational dissolution explain survivor homogenization,
temporally concentrated exit, and fragmentation into cohesive subgroups in
different vocabularies.  Their recurrence is the fact to be explained.  A
finite condition structure with a minimal dynamic extension generates all
three, each in its own regime, with no actor responding to another; we then
determine which conclusions survive when the restrictions defining those
regimes are relaxed.  With two requirement masks the union class is the only
mediator between the pure types and the signature with the largest closed
neighbourhood; under positive endpoint loads it is also the unique maximizer
of deficient load.  The independent derivations of fragmentation and
homogenization therefore both concern that class.  With more than two masks
the three roles need not coincide.  Forced coincidence admits an exact
characterization for every number of masks: it holds if and only if the
requirement-derived join lies in the realized family and separates the
endpoints.  The proportion of family--pair instances meeting this condition
vanishes under a uniform counting measure as the number of masks increases,
and, at one actor per signature, an edge density above a stated bound
excludes it.  A common substrate does not make the mechanisms identical:
selective attrition can hold while correlated collapse fails, and a change
in exit dependence can change the population outcome while every actor's
marginal lifetime law is unchanged.
\end{abstract}

\clearpage
\setcounter{tocdepth}{1}
\tableofcontents
\clearpage

\input{sections/01_introduction}

\input{sections/02_substrate}

\input{sections/03_dynamics}

\input{sections/04_generation}

\input{sections/05_recovery}

\input{sections/06_boundaries}

\input{sections/07_observables}

\input{sections/08_discussion}

\input{sections/09_conclusion}

\appendix
\input{sections/A_appendix}

\bibliographystyle{plainnat}
\bibliography{refs}
\end{document}

%% file: numbers.tex
\newcommand{\pvBalanceReadout}{0.9102}  %
\newcommand{\pvStratPairsRThreeSThree}{15}  %
\newcommand{\pvStratForcedRThreeSThree}{6}  %
\newcommand{\pvStratPairsRThreeSFour}{48}  %
\newcommand{\pvStratForcedRThreeSFour}{12}  %
\newcommand{\pvStratPairsRFourSThree}{151}  %
\newcommand{\pvStratForcedRFourSThree}{25}  %
\newcommand{\pvStratPairsRFourSFour}{1416}  %
\newcommand{\pvStratForcedRFourSFour}{144}  %
\newcommand{\pvNonInvariantIsoClasses}{17}  %
\newcommand{\pvInvGoodCount}{14289}  %
\newcommand{\pvCoordSupportEnum}{60884}  %
\newcommand{\pvInvFailCount}{125614}  %
\newcommand{\pvCapViolPAHZero}{1710}  %
\newcommand{\pvCapGoodHZero}{9891}  %
\newcommand{\pvCapGenUnitCases}{1323}  %
\newcommand{\pvObservablesCollapsed}{32341}  %

\newcommand{\pvCapGoodComposition}{1308}  %
\newcommand{\pvCapViolSigmaComposition}{469}  %
\newcommand{\pvCapGenInstances}{3969}  %
\newcommand{\pvCapGenViolSigma}{0}  %
\newcommand{\pvCapGenViolUnit}{415}  %
\newcommand{\pvCapGenTight}{217}  %

\newcommand{\pvAsaCitations}{6{,}000}  %
\newcommand{\pvAsaHypothesisStudies}{321}  %
\newcommand{\pvAsaTemporalStudiesWord}{one}  %
\newcommand{\pvHaudekDecreased}{42.5}  %
\newcommand{\pvHaudekIncreased}{57.5}  %

%% file: sections/01_introduction.tex
\section{Introduction}\label{sec:intro}

Closely related morphologies of withdrawal, fragmentation, and homogenization
have been explained many times, in separate theoretical vocabularies.  Turnover
arrives in clusters among employees who occupy similar informal roles
\citep{krackhardt1986snowball,porter2021contagion}.  Participation collapses
once a chain of individual thresholds is crossed
\citep{granovetter1978threshold,macy2020threshold}.  Structurally unbalanced
relations resolve toward balance \citep{cartwright1956structural,pham2022triad}.
Withdrawal proceeds through a chain of intermediate decisions
\citep{mobley1977intermediate,lee2017decade}.  Organizations become more
homogeneous through attraction, selection, and attrition
\citep{schneider1987asa,oh2018birds}.

These mechanisms are established; the attraction--selection--attrition
framework alone has been cited in more than $\pvAsaCitations$ papers, with
$\pvAsaHypothesisStudies$ studies using it as the basis of a hypothesis
\citep{vaniddekinge2026asa}.  The accounts developed in separate literatures,
with different explanatory variables and time scales, yet the morphologies they
explain resemble one another.  Three recur:

\begin{itemize}
  \item[(M1)] \emph{Survivor homogenization.}  Survivors end up more alike
    than the population they were drawn from.
  \item[(M2)] \emph{Temporal concentration.}  Exits concentrate in time.
  \item[(M3)] \emph{Fragmentation.}  The group splits into internally cohesive
    fragments.
\end{itemize}

What recurs is the theoretical expectation, not a settled regularity: in the
\pvAsaTemporalStudiesWord{} study the review just cited identifies as
examining homogeneity over time, within-firm variance fell in
$\pvHaudekDecreased\%$ of cases and rose in $\pvHaudekIncreased\%$.  The
mechanisms diverge; the morphologies recur.  This asymmetry is the fact
requiring explanation.  A shared source may lie in the organizations
themselves, in transmission among the literatures, or nowhere, the recurrence
being contingent.  Transmission and contingency cannot be excluded by
construction; a source of the first kind can be exhibited, and we construct one.

The structural substrate is inherited.  A \emph{finite condition
structure}---actors depending on distinct bundles drawn from a common finite set
of conditions that sustain their relations---is the object of our earlier static theory
\citep{imptheom}, which states what such a structure forces once positions are
fixed.  We place a minimal dynamic extension on the same structure, adding no new
primitive to $\mathcal S$ and carrying the veto semantics over as an operator on the surviving
population, and show that the resulting system generates (M1)--(M3); this is a
sufficiency result.  A second question follows: does the shared substrate merely
permit several mechanisms, or does it force roles defined independently of one
another---mediation between the pure types, closed-neighbourhood size, and
deficient load---onto one position?  With two requirement masks it does, and
the independent derivations of fragmentation and homogenization both concern the
class the requirement lattice singles out.  With more masks it is no longer
forced.  Forced coincidence admits an exact characterization for every number of
masks, the join lying in the realized family and separating the endpoints, and
the proportion of family--pair instances meeting it vanishes as the number of
masks increases.

\subsection*{Results}

\begin{enumerate}
  \item \emph{Generation.}  A finite condition structure with a minimal dynamic
    extension generates (M1)--(M3), with no observer in the state space and no
    assumption that actors respond to one another, each in a regime of its own:
    (M1) from concentration at the persistence maximum
    (Lemma~\ref{lem:persistmax}), with the mixed class over-represented among
    the exits at every finite horizon (Proposition~\ref{prop:enrichment}); (M2) from route
    selection at $\beta=1$ (Proposition~\ref{prop:routeselector}); (M3) from
    thin-bridge separation (Theorem~\ref{thm:thinbridge}).  The exit-dependence
    parameter $\beta$ distinguishes the first two on the frozen slice
    (Section~\ref{sec:recovery}); what the other coordinates do to these results
    is stated in Section~\ref{sec:coordinates}.
  \item \emph{Alignment at two masks.}  With two requirement masks the union class
    is at once the only mediator between the pure types, the signature with the
    largest closed neighbourhood, and, under positive endpoint loads, the unique
    maximizer of deficient load (Corollary~\ref{cor:forcedcoincidence}).  The
    derivations of (M3) and (M1) are independent---one uses the mediating
    role (Corollary~\ref{cor:realizedmediator}), the other the load ordering
    (Proposition~\ref{prop:supporthazard})---and the lattice makes both apply
    to the same class.  At three or more masks the alignment is not forced.
  \item \emph{Boundary.}  Forced coincidence is characterized by the
    requirement-derived join together with local connectivity, becomes
    asymptotically negligible under a uniform counting measure as the number of
    masks grows, and implies a cap on an observable density; the same section
    locates the boundary of the dynamic results in the coordinates those results
    hold fixed (Section~\ref{sec:boundaries}).
  \item \emph{Non-identity.}  Sharing a generator does not merge the
    mechanisms: selective attrition holds while correlated collapse fails in one
    regime (Proposition~\ref{prop:nonidentity}), and both hold with different
    outcomes in another (Proposition~\ref{prop:divergentendpoints}).
  \item \emph{Observable implications.}  Distinguishing among latent
    structural possibilities requires information the projected graph does not
    preserve; some projections retain enough to exclude classes of latent
    structure (Corollary~\ref{cor:inverseexclusion},
    Section~\ref{sec:observables}).
\end{enumerate}

%% file: sections/02_substrate.tex
\section{Structural substrate}\label{sec:substrate}

\subsection{The finite condition structure}

\begin{definition}[Finite condition structure]\label{def:substrate}
A \emph{finite condition structure} is a tuple
\[
  \mathcal{S}=\bigl(P,\;G,\;a,\;L,\;\{f_g\}_{g\in G}\bigr),
\]
where $P$ is a finite set of actors with $N:=|P|\ge 2$; $G$ is a finite set of
gain axes; $a\colon P\times G\to\{0,1\}$ is the minimum-condition function;
$L$ is a position space carrying no topology, metric, or order; and each
$f_g\colon L\times L\to\{0,1\}$ is a symmetric compatibility function.
\end{definition}

The axis set $G$ is an axiomatic input rather than something derived.  The map $a$ is fixed once and for all and is \emph{partner-independent}:
whether $i$ requires axis $g$ does not depend on whom $i$ is paired with.
Consequently the \emph{requirement signature}
\[
  s_i:=\{\,g\in G\;:\;a(i,g)=1\,\}\subseteq G
\]
is a derived object rather than a primitive.

A configuration is a map $E\colon P\to L$.  It is not part of $\mathcal{S}$.

\subsection{Inherited statements}

\begin{description}
  \item[(S1) Minimum condition.]  $a(i,g)=1$ means that axis $g$ is a minimum
    condition for $i$: failure on $g$ is sufficient to make the relation
    non-viable for $i$ \citep{imptheom}.
  \item[(S2) Veto asymmetry.]  Severance is disjunctive and confirmation is
    conjunctive.  Incompatibility on axis $g$ at the given configuration, that
    is $f_g(E(i),E(j))=0$, severs the relation when $a(i,g)=1$ \emph{or}
    $a(j,g)=1$.  An axis $g\in s_i\cap s_j$ is \emph{shared} by $i$ and $j$ as a
    property of $\mathcal{S}$ alone; it counts as \emph{confirmed} only when in
    addition $f_g(E(i),E(j))=1$, which depends on $E$.
\end{description}

\noindent
(S2) draws candidate veto axes from the union $s_i\cup s_j$ and candidate axes for bilateral confirmation from the intersection $s_i\cap s_j$; whether a candidate is compatible still depends on $f_g(E(i),E(j))$.

Two claims are made under the word \emph{inheritance}: (I) no new primitive is added to the tuple $\mathcal{S}$ of Definition~\ref{def:substrate}; (II) the minimum-condition rule is taken as the removal semantics of the dynamic model, realized in Section~\ref{sec:dynamics} as an operator on subsets of $P$.

%% file: sections/03_dynamics.tex
\section{Minimal dynamic extension}\label{sec:dynamics}

\subsection{State}

\begin{definition}[Dynamic state]\label{def:state}
Fix a finite condition structure $\mathcal{S}$.  The
dynamic state at time $t$ consists of
\[
  R(t)\subseteq R_0\subseteq P,
  \qquad
  \varphi_{ig}(t)\in\{0,1\}\ \ (i\in P,\ g\in G),
\]
where $R(t)$ is the surviving population, non-increasing in $t$, and
$\varphi_{ig}(t)$ records whether axis $g$ is currently satisfied for actor
$i$.  We call $\varphi$ the \emph{satisfaction state}.
\end{definition}

\noindent
The dynamics replace the pairwise compatibility $f_g(E(i),E(j))$ of
Section~\ref{sec:substrate} by the unary state $\varphi_{ig}$, and use neither
$f_g$ nor $E$ thereafter.  Since $\varphi_{ig}$ is a state variable only for the
incidences that exist, we set $\varphi_{ig}:=0$ when $g\notin s_i$; for $g\in G$
and $t\ge 0$ we write $Q_g(t):=\{\,i\in P:\varphi_{ig}(t)=1\,\}$ for the
\emph{satisfaction support} of $g$.

\subsection{The dynamic law}

The law fixes how conditions fail and are restored, and when a deficient actor
leaves; carrying (S1) over to removal is claim (II) of
Section~\ref{sec:substrate}, taken rather than inherited.

Write $s_i\subseteq G$ for the requirement signature of $i$ and
\[
  d_i(t):=\bigl|\{\,g\in s_i\;:\;\varphi_{ig}(t)=0\,\}\bigr|
\]
for its current number of deficient axes.  For an axis $g$, let the satisfied
and deficient actors on $g$ among survivors be
\[
  P_g(t):=Q_g(t)\cap R(t),
  \qquad
  D_g(t):=\{\,i\in R(t)\;:\;g\in s_i,\ \varphi_{ig}(t)=0\,\}.
\]
The clocks act on these sets.  They differ from the satisfaction support
$Q_g(t)$ by survivorship: exit does not alter $\varphi$, so a departed actor
keeps its satisfaction state, remaining in $Q_g(t)$ while leaving $P_g(t)$.

\begin{definition}[Transition intensities]\label{def:law}
Fix parameters
\[
  \lambda>0,
  \qquad \nu,\rho\ge 0,
  \qquad \beta,\rho_S\in[0,1].
\]
Conditional on an initial state $(R_0,\varphi(0))$, the process is a
continuous-time Markov jump process with the following independent exponential
clocks.

\emph{Loss} $\varphi:1\to 0$.
\begin{itemize}
  \item Shared: one clock per axis $g$ with $P_g(t)\neq\emptyset$, at rate
    $\rho_S\,\nu$, setting $\varphi_{ig}\to 0$ for \emph{every}
    $i\in P_g(t)$ at once.
  \item Idiosyncratic: one clock per incidence $(i,g)$ with $i\in P_g(t)$, at
    rate $(1-\rho_S)\,\nu$, setting $\varphi_{ig}\to 0$.
\end{itemize}

\emph{Supply} $\varphi:0\to 1$.  One clock per incidence $(i,g)$ with
$i\in D_g(t)$, at rate $\rho$.  There is no shared component.

\emph{Exit} $i\in R(t)\to i\notin R(t)$.
\begin{itemize}
  \item Idiosyncratic: one clock per deficient incidence $(i,g)$, at rate
    $(1-\beta)\lambda$, removing $i$.  Aggregating over $g$, actor $i$ leaves
    at rate $(1-\beta)\lambda\,d_i(t)$.
  \item Common: one clock per axis $g$ with $D_g(t)\neq\emptyset$, at rate
    $\beta\lambda$, removing \emph{every} actor in $D_g(t)$ at once.
\end{itemize}
\end{definition}

\noindent
Throughout, \emph{shared} refers to the loss clock scaled by $\rho_S$ and
\emph{axis-common} to the exit clock scaled by $\beta$; the two act on
different transitions.

\subsection{One scale and four free coordinates}

Definition~\ref{def:law} specifies five parameters.  After rescaling time, four
dimensionless coordinates remain.

\begin{remark}[$\lambda$ is a gauge]\label{rem:gauge}
The scale $\lambda$ enters only the exit intensities.  Rescaling time by
$\lambda$ leaves the law invariant provided $\nu$ and $\rho$ are rescaled with
it, so trajectories depend on
$(\beta,\rho_S,\nu/\lambda,\rho/\lambda)$ alone.  We may set $\lambda=1$ when counting free
coordinates, and then write them as $(\beta,\nu,\rho_S,\rho)$.  When $\nu=0$ the
coordinate $\rho_S$ has no effect, since there are no losses for its clock to
drive.
\end{remark}

\begin{center}
\begin{tabular}{@{}llp{0.55\linewidth}@{}}
\toprule
Coordinate & Transition & What it controls \\
\midrule
$\beta$   & exit  & dependence among exits: idiosyncratic incidence clocks
                    versus axis-common clocks \\
$\nu$     & loss  & rate at which satisfied incidences become deficient \\
$\rho_S$  & loss  & shared versus idiosyncratic loss \\
$\rho$    & supply   & irreversible forcing versus resupply \\
\bottomrule
\end{tabular}
\end{center}

\subsection{How far the inherited axiom is preserved}

When $\rho=0$ a deficiency is permanent and the exit clock fires almost surely,
so $\rho=0$ gives eventual forcing; the dynamics relax the timing, not the
forcing.  When $\rho>0$ a deficient actor may be resupplied first, and
Section~\ref{sec:coordinates} gives the probability of permanent survival.
Setting $\nu=\rho=0$ gives the frozen slice, on which the deficiency incidence
is fixed and inherited static statements transfer.  Trajectory-level stability
is not asserted: even at $\rho=0$ an actor can be deficient and still present
until its exit clock fires.  Only $\beta$ is varied by the results below, which
determine the morphologies at both extremes (Section~\ref{sec:recovery}); its loss-side
counterpart $\rho_S$ is treated symmetrically in
Lemma~\ref{lem:marginalinvariance} and degenerates on the slice
(Remark~\ref{rem:gauge}).  Section~\ref{sec:coordinates} treats the other three
coordinates, among them the regime $\nu>0$ with $\rho=0$, which preserves
forcing while lying outside the slice.

%% file: sections/04_generation.tex
\section{The generating structure}\label{sec:generation}

\subsection{What the commonality parameters leave fixed}

\begin{lemma}[Marginal invariance]\label{lem:marginalinvariance}
Fix an axis $g$.  Each satisfied incidence $(i,g)$ with $i\in P_g(t)$ is lost at
total rate $(1-\rho_S)\nu+\rho_S\nu=\nu$, and each deficient incidence $(i,g)$
contributes total rate $(1-\beta)\lambda+\beta\lambda=\lambda$ to the exit of
$i$.  Hence $\rho_S$ and $\beta$ may alter the dependence structure, and hence joint
trajectories, without altering the marginal incidence rates.
\end{lemma}

\begin{proof}
Immediate from Definition~\ref{def:law}: the shared clock of an axis fires at
rate $\rho_S\nu$ and affects every actor in $P_g(t)$, so it contributes
$\rho_S\nu$ to each; the idiosyncratic clock of the incidence fires at
$(1-\rho_S)\nu$.  The exit case is identical with $\lambda$ in place of $\nu$.
\end{proof}

\noindent Section~\ref{sec:coordinates} strengthens this from rates to laws, for all $\nu,\rho\ge0$.

\subsection{Two graphs, and the difference between them}

\begin{definition}[Cohesion graph]\label{def:cohesion}
For $i,j\in R(t)$ with $i\neq j$, put $\{i,j\}\in E_C(t)$ when there exists
$g\in s_i\cap s_j$ with $\varphi_{ig}(t)=\varphi_{jg}(t)=1$.  Write $G_C(t):=(R(t),E_C(t))$.
\end{definition}

\noindent
Adjacency requires a condition both parties require and both currently
satisfy, using the satisfaction state of Section~\ref{sec:dynamics} rather
than $f_g$.  The generating argument is about which adjacencies the
requirement structure permits at all, a separate and static object.

Group the axes by who requires them: let $g\sim g'$ when
$\{i:a(i,g)=1\}=\{i:a(i,g')=1\}$ and call the resulting blocks
$\mathcal{G}_1,\dots,\mathcal{G}_r$ \emph{requirement masks}.  Requirement is
constant on each block, so the signature of Section~\ref{sec:substrate}
satisfies $s_i=\bigcup_{m}\mathcal{G}_m$ over the masks $m$ that $i$ requires;
the two descriptions determine each other, and from here on $s_i\subseteq[r]$
is the mask form, $r$ the number of masks.  Only nonempty signatures are
considered, giving
\[
  \Sigma_r:=2^{[r]}\setminus\{\emptyset\}.
\]
Joins are taken in the ambient Boolean lattice $2^{[r]}$, of which $\Sigma_r$ is
the join-closed subposet of nonempty elements; disjoint signatures meet at
$\emptyset$, no meet is used below, and \emph{requirement lattice} names
$2^{[r]}$.

\begin{definition}[Ambient structural signature graph]\label{def:sigraph}
Let $H_r:=(\Sigma_r,E_\cap)$ be the simple graph with $\{s,u\}\in E_\cap$
exactly when $s\neq u$ and $s\cap u\neq\emptyset$.
\end{definition}

\noindent
$H_r$ is the subset intersection graph on $\Sigma_r$ \citep{scheinermanwest1989},
in the vocabulary standard for intersection graphs of set families
\citep{mckeemcmorris1999}.  Its vertices are the admissible signature types, not the classes realized in the current population; a vertex may be realized by no actor at all.

\begin{lemma}[Adjacency requires signature intersection]\label{lem:signatureintersection}
If $s_i\cap s_j=\emptyset$ then $\{i,j\}\notin E_C(t)$ for every $t$.
\end{lemma}

\begin{proof}
Definition~\ref{def:cohesion} requires some $g\in s_i\cap s_j$, and there is
none.
\end{proof}

\begin{remark}[Permission versus realization]\label{rem:permission}
An edge of $H_r$ becomes an edge of $G_C(t)$ only if both parties survive and
both satisfy an axis of a shared mask, and outside the slice of
Definition~\ref{def:frozen} the axes of a mask need not move together.  Every
structural statement below about mediation, closed-neighbourhood size, and separation
is about $H_r$, not about paths simultaneously realized in $G_C(t)$.
\end{remark}

\begin{remark}[Where non-transitivity can come from]\label{rem:nontransitivity}
Fix a single axis $g$.  Adjacency through $g$ holds for every pair in the
satisfaction support $Q_g(t)$, so restricted to the actors who require and
satisfy $g$ it is transitive.  A failure of transitivity in $G_C(t)$ can
arise only through the disjunction over $g$ in Definition~\ref{def:cohesion},
once information from several axes is combined.
\end{remark}

\subsection{Two monotone quantities}

\begin{definition}[Closed neighbourhood]\label{def:connectivecapacity}
For $s\in\Sigma_r$ let
\[
  \mathcal{C}(s):=\{\,u\in\Sigma_r\ :\ s\cap u\neq\emptyset\,\}
\]
be the closed neighbourhood of $s$ in $H_r$, that is, the signatures not
excluded from adjacency with $s$ by
Lemma~\ref{lem:signatureintersection}.  Write $c(s):=|\mathcal{C}(s)|$.
\end{definition}

\begin{lemma}[The closed neighbourhood is support-monotone]\label{lem:capacitymonotone}%
\footnote{Degree monotonicity in the intersection graph of lattice ideals
\citep{afkhami2015} is the analogue of this lemma.  We claim no novelty for it.}
If $s\subseteq s'$ then $\mathcal{C}(s)\subseteq\mathcal{C}(s')$.  Moreover
$\mathcal{C}([r])=\Sigma_r$, and $[r]$ is the unique maximizer of $c(\cdot)$.
\end{lemma}

\begin{proof}
If $s\cap u\neq\emptyset$ and $s\subseteq s'$ then $s'\cap u\neq\emptyset$,
which gives the inclusion.  Every $u\in\Sigma_r$ meets $[r]$, so
$\mathcal{C}([r])=\Sigma_r$.  If $s\subsetneq[r]$, pick $m\in[r]\setminus s$;
then $\{m\}\in\mathcal{C}([r])\setminus\mathcal{C}(s)$, so the inclusion is
strict.
\end{proof}

\begin{lemma}[Closed form]\label{lem:capacitycount}
$c(s)=2^{r}-2^{\,r-|s|}$.  In particular $c(s)$ depends on
$s$ only through $|s|$, and is strictly increasing in it.
\end{lemma}

\begin{proof}
The signatures disjoint from $s$ are the nonempty subsets of $[r]\setminus s$,
of which there are $2^{\,r-|s|}-1$.  Subtracting from $|\Sigma_r|=2^{r}-1$
gives the count.
\end{proof}

The second quantity depends on the dynamics and is defined on a restricted
set of states, the slice introduced next.

\begin{definition}[Shared-state frozen slice]\label{def:frozen}
The \emph{shared-state frozen slice} consists of the states with $\nu=\rho=0$
and $\varphi_{ig}=\varphi_{jg}$ for every axis $g$ and every pair $i,j$
requiring $g$.  Since $\nu=\rho=0$, satisfaction does not move, so the common
value persists and we write $\varphi_g$ for it.  Put
\[
  \ell_m:=\bigl|\{\,g\in\mathcal{G}_m:\varphi_g=0\,\}\bigr|,
  \qquad
  L_s:=\sum_{m\in s}\ell_m .
\]
\end{definition}

\begin{proposition}[Support-monotone frozen hazard]\label{prop:supporthazard}
On the shared-state frozen slice, the lifetime of an actor with signature $s$
satisfies
\[
  T_s\sim\mathrm{Exp}(\lambda L_s),
  \qquad\text{equivalently}\qquad
  T_s\overset{d}{=}\min_{m\in s}T'_m,
\]
where, \emph{for each actor separately and with its own family}, the $T'_m$ are
independent, $T'_m\sim\mathrm{Exp}(\lambda\ell_m)$, and
$T'_m=\infty$ when $\ell_m=0$; if $L_s=0$ then $T_s=\infty$.  Consequently
$s\subseteq s'$ gives $L_s\le L_{s'}$ and $T_{s'}\le_{\mathrm{st}}T_s$.  The
marginal distribution does not depend on $\beta$.

If moreover $\beta=0$, the surviving exit clocks of Definition~\ref{def:law} are
indexed by deficient incidences $(i,g)$, and these index sets are disjoint across
actors, so the lifetimes of distinct actors are mutually independent.  For $\beta>0$ actors deficient on a common axis share the axis-common clock, and their lifetimes are positively dependent.
\end{proposition}

\begin{proof}
See Appendix~\ref{app:proofs-generation}.
\end{proof}

The closed neighbourhood is set-theoretic and does not depend on the
parameters; the load needs Definition~\ref{def:frozen}.

\subsection{Forcing at \texorpdfstring{$r=2$}{r=2}}

\begin{lemma}[Unique mediator at $r=2$]\label{lem:uniquemediator}
Let $r=2$ and write $\Sigma_2=\{A,AB,B\}$ with $A=\{1\}$, $B=\{2\}$ and
$AB=\{1,2\}$.  Then $AB$ is the only signature meeting both $A$ and $B$; every
$A$--$B$ path in $H_2$ contains $AB$, so $\{AB\}$ is an $A$--$B$ separator in
$H_2$; and $AB=[2]$ is the top element of the signature lattice.
\end{lemma}

\begin{proof}
A signature meeting both $A$ and $B$ must contain $1$ and $2$, hence equals
$[2]=AB$.  Since $A\cap B=\emptyset$, the edge $\{A,B\}$ is absent from $H_2$,
and $AB$ is the only other vertex, so every $A$--$B$ path runs through it.  That
$AB=[2]$ is the top element is immediate.
\end{proof}

\noindent
Write $\ell_A:=\ell_1$ and $\ell_B:=\ell_2$.

\begin{corollary}[Realized consequence]\label{cor:realizedmediator}
Any realized cohesion path between an actor with signature $A$ and an actor
with signature $B$, if one exists, passes through an actor with signature $AB$.
\end{corollary}

\begin{proof}
By Lemma~\ref{lem:signatureintersection} a realized path projects to a path in
$H_2$ between $A$ and $B$, and Lemma~\ref{lem:uniquemediator} applies.
\end{proof}

\begin{corollary}[The three roles align at $r=2$]\label{cor:forcedcoincidence}
Let $r=2$.  The unique structural mediator between the pure types is $AB=[2]$,
and $[2]$ is simultaneously the unique maximizer of $c(\cdot)$
(Lemma~\ref{lem:capacitymonotone}) and a maximizer of $L_\cdot$
(Proposition~\ref{prop:supporthazard}).  For $r\ge3$ the forcing can fail: the
chain $\{1\},\{1,2\},\dots,\{r-1,r\},\{r\}$ connects the pure types without
using $[r]$, and if $\ell_1,\ell_r>0$ while $\ell_2=\dots=\ell_{r-1}=0$ then
every internal signature $s$ on that chain satisfies $L_s<L_{[r]}$.  The witness
is one configuration; it shows that the coincidence is no longer forced at
$r\ge3$, not that it never occurs there.
\end{corollary}

\noindent
Thus at $r=2$ the three roles, fixed independently of one another, align on a
single class; this is stronger than the separator condition
Section~\ref{sec:boundaries} calls forced coincidence.

At $r=2$, $L_{AB}=L_A+L_B$ by Definition~\ref{def:frozen}, so the two
maximality claims of Corollary~\ref{cor:forcedcoincidence} differ in strength:
the maximizer of $c(\cdot)$ is unique outright, the load maximizer need not be.  If
$\ell_A=0$ then $L_{AB}=L_B$ and the pure type $B$ ties with the union class.
Under the positive endpoint-load condition $\ell_A,\ell_B>0$, the
condition that selective attrition assumes in Section~\ref{sec:recovery}, the
sum is strict on both sides and $[2]$ is the \emph{unique} maximizer of $L_\cdot$ as well.
Statements below asserting uniqueness of the load maximizer are made under that
condition.

Both monotonicities hold at every $r$; only at $r=2$ do a set-theoretic
maximizer of $c(\cdot)$, a load maximizer depending on the frozen dynamics, and
a mediator forced by the lattice coincide on one class.
Section~\ref{sec:recovery} uses the two properties separately, the mediating
side in Theorem~\ref{thm:thinbridge} and the load ordering in
Proposition~\ref{prop:enrichment}, and the alignment adds that the two
derivations both concern the same class.  At $r\ge3$ the chain of
Corollary~\ref{cor:forcedcoincidence} distributes mediation across proper mixed
signatures, and what is lost is the forced alignment, not either monotonicity.

%% file: sections/05_recovery.tex
\section{Generated outcomes and distinct mechanisms}\label{sec:recovery}

\subsection{Recovered mechanisms}

We distinguish eight mechanisms recoverable inside the present model.

\begin{center}
\begin{tabular}{@{}p{0.20\linewidth}p{0.27\linewidth}p{0.33\linewidth}@{}}
\toprule
Mechanism & Regime or condition & What it asserts \\
\midrule
\multicolumn{3}{@{}l}{\emph{Latent mechanisms}}\\
Brokerage & structural; $r=2$, both pure types present
  & the mixed class is the only mediator \\
Broker burden & frozen slice; $L_{AB}=L_A+L_B$
  & it carries the summed load \\
Overlapping membership & structural
  & closed-neighbourhood size grows with the signature \\
Selective attrition & frozen slice, $\ell_A,\ell_B>0$
  & it is stochastically shortest lived \\
Bridge depletion & frozen slice, $\beta=0$; bridge thin in persistence
  & it is exhausted while both pure classes survive \\
Fragmentation & thin-bridge separation
  & survivors occupy cohesive components \\
Correlated collapse & $\beta>0$; an axis with two deficient actors
  & an axis-common removal channel is active \\
\midrule
\multicolumn{3}{@{}l}{\emph{Projection mechanism}}\\
Apparent cohesion & concerns the observation map
  & a projection can present survivors as cohesive \\
\bottomrule
\end{tabular}
\end{center}

\noindent
The entries are predicates about the model: selective attrition orders
lifetimes, correlated collapse says a multi-actor channel is active, and the
first three are relational premises the later results use.

\subsection{Where each mechanism comes from}

Each entry arises from the substrate of Section~\ref{sec:substrate} under the
dynamics of Section~\ref{sec:dynamics}, at the regime named in the table;
Section~\ref{sec:coordinates} states which entries survive when that regime is
left, and which couplings the slice suppresses.  Brokerage follows from
Lemma~\ref{lem:uniquemediator}, in realized form from
Corollary~\ref{cor:realizedmediator}; overlapping membership from
Lemma~\ref{lem:capacitymonotone}; broker burden and selective attrition from
Proposition~\ref{prop:supporthazard} on the slice of
Definition~\ref{def:frozen}; correlated collapse from the axis-common exit
channel of Definition~\ref{def:law}, active when $\beta>0$ and some axis has at
least two deficient actors.  Bridge depletion and fragmentation need more than a
regime label.

\subsection{Bridge depletion and fragmentation}

\begin{definition}[Persistence scale]\label{def:persistence}
On the shared-state frozen slice let $n_s$ be the number of actors with
signature $s$.  The \emph{persistence scale} of the class is
\[
  W_s:=\frac{\log n_s}{\lambda L_s},
  \qquad W_s:=\infty \ \text{ when } L_s=0 .
\]
\end{definition}

\begin{remark}[What the scale leaves to the load]\label{rem:leadingorder}
At $\beta=0$ let $M_s$ be the time at which the last actor of signature $s$
exits.  By Proposition~\ref{prop:supporthazard} the $n_s$ lifetimes are independent
$\mathrm{Exp}(\lambda L_s)$, so $\lambda L_sM_s-\log n_s$ converges in
distribution to a standard Gumbel law as $n_s\to\infty$, and
$M_s=W_s+O_P\bigl((\lambda L_s)^{-1}\bigr)$.  In the leading term, load and
class size trade off logarithmically: a class with twice the load needs the
square of the class size to reach the same $W$.  Where $\log n_s$ stays bounded the leading
terms do not separate and the correction decides.  At $n_s=1$ for every class the
leading terms vanish and $M_s=T_s$, leaving the load ordering of
Proposition~\ref{prop:supporthazard}.
\end{remark}

\begin{theorem}[Thin-bridge separation]\label{thm:thinbridge}
Consider a two-mask shared-state frozen slice in which only the signatures
$A$, $AB$ and $B$ are present.  Assume: (i) $\beta=0$; (ii) $\ell_A,\ell_B\ge1$,
and these remain fixed as the population grows; (iii) each requirement mask also retains
at least one nondeficient axis, so the surviving pure classes are internally
cohesive; and (iv) the mixed class is thin in the persistence scale of
Definition~\ref{def:persistence}, that is
$\min\{W_A,W_B\}-W_{AB}\to\infty$ as the population grows.  Then
the probability that every actor in the mixed class has exited while at least
one pure-$A$ and one pure-$B$ actor remain tends to one.  At that time the
surviving $A$ and $B$ actors lie in distinct components of $G_C(t)$, each
internally cohesive.
\end{theorem}

\begin{proof}
See Appendix~\ref{app:proofs-recovery}.
\end{proof}

\begin{remark}[What thinness measures]\label{rem:thinness}
Since $L_{AB}=\ell_A+\ell_B$, hypothesis (iv) reads
\[
  n_{AB}\ \ll\ \min\bigl\{\,n_A^{(\ell_A+\ell_B)/\ell_A},\ \,
                        n_B^{(\ell_A+\ell_B)/\ell_B}\,\bigr\},
\]
where $x\ll y$ abbreviates $x/y\to0$ along the sequence, which is what the
divergence in (iv) delivers.  Both exponents exceed one, so a bridge may grow
polynomially in the pure classes and still be exhausted first.  Hypothesis (iv)
also supplies the divergence of the pure classes, since $W_{AB}\ge0$ and
$\ell_A,\ell_B$ are fixed by (ii), so $W_A,W_B\to\infty$ forces
$n_A,n_B\to\infty$.
\end{remark}

\subsection{Selective and correlated exit}

Selective attrition and correlated collapse are derived on the same slice
from the same loads; a regime in which one holds and the other fails shows
that they remain distinct.

\begin{proposition}[Selective attrition without correlated exit]\label{prop:nonidentity}
On the shared-state frozen slice of Definition~\ref{def:frozen} with $\beta=0$,
$\ell_A,\ell_B\ge1$, and all three signatures of $\Sigma_2$ present in the
population, selective attrition holds and correlated collapse fails.  The two
predicates are therefore distinct.
\end{proposition}

\begin{proof}
See Appendix~\ref{app:proofs-recovery}.
\end{proof}

\begin{corollary}[Finite-horizon form]\label{cor:finitehorizon}
Under the hypotheses of Proposition~\ref{prop:nonidentity}, for every $t>0$ an
actor in the mixed class has strictly larger exit probability by time $t$
than an actor in either pure class:
$1-e^{-\lambda(\ell_A+\ell_B)t}>1-e^{-\lambda\ell_A t}$ and likewise for
$\ell_B$.
\end{corollary}

\begin{proof}
Immediate from Proposition~\ref{prop:supporthazard}.
\end{proof}

\begin{proposition}[Compositional enrichment]\label{prop:enrichment}
Assume the hypotheses of Proposition~\ref{prop:nonidentity}.  Fix the reference
time, write $\Sigma$ for the signature of a uniformly drawn actor from the
population at risk, and let
\[
  \pi_s:=\Pr\bigl(\Sigma=s\mid\text{at risk at the reference time}\bigr),
  \qquad s\in\{A,AB,B\},
\]
be the signature composition of the population at risk \emph{at that time}, with
$0<\pi_{AB}<1$; write $p_s(t):=\Pr(T_s\le t)$.  Then for every $t>0$
\begin{multline*}
  E_{AB}(t)\;:=\;\frac{\Pr\bigl(\Sigma=AB\mid T\le t\bigr)}{\pi_{AB}}\\
  \;=\;\frac{p_{AB}(t)}
            {\pi_A p_A(t)+\pi_{AB}p_{AB}(t)+\pi_B p_B(t)}\;>\;1 .
\end{multline*}
The mixed class is over-represented among the exits by time $t$ relative to its
share of the population at risk.
\end{proposition}

\begin{proof}
See Appendix~\ref{app:proofs-recovery}.
\end{proof}

\begin{lemma}[Concentration at the persistence maximum]\label{lem:persistmax}
Consider the shared-state frozen slice of Definition~\ref{def:frozen} (so
$\nu=\rho=0$ and $\varphi_{ig}=\varphi_{jg}$ for every axis $g$ shared by two
actors requiring it), with $\beta=0$.  Let $S\subseteq\Sigma_r$ be
the set of signatures present in the population at the reference time;
Section~\ref{sec:boundaries} calls this set the realized family, and for
$s\in S$ let $n_s$ be the number of actors with signature $s$.  By
Proposition~\ref{prop:supporthazard} the lifetimes
$T_s\sim\mathrm{Exp}(\lambda L_s)$ are mutually independent across actors.

Let $\mathcal A\subseteq S$ be nonempty with $L_s>0$ for every
$s\notin\mathcal A$, and write $\mathrm{LAST}(\mathcal A)$ for the event that every actor
outside $\mathcal A$ has exited while every $s\in\mathcal A$ still has at least
one surviving actor.  Then for every $x>0$
\begin{equation}\label{eq:persistbound}
  \Pr\bigl[\mathrm{LAST}(\mathcal A)\bigr]\;\ge\;
  1-\sum_{s\notin\mathcal A}n_s\,e^{-\lambda L_sx}
   -\sum_{s\in\mathcal A}\exp\bigl(-n_s\,e^{-\lambda L_sx}\bigr),
\end{equation}
and on $\mathrm{LAST}(\mathcal A)$ the surviving signature set at the time the last actor
outside $\mathcal A$ exits equals $\mathcal A$.

Now consider a sequence of populations in which $S$ and the loads
$\{L_s\}_{s\in S}$ are held fixed and only the multiplicities vary.  If
$\mathcal A$ separates in the persistence scale of
Definition~\ref{def:persistence}, that is if
\[
  \min_{s\in\mathcal A}W_s-\max_{s\notin\mathcal A}W_s\longrightarrow\infty ,
\]
then $\Pr[\mathrm{LAST}(\mathcal A)]\to1$.  Every separating $\mathcal A$ contains
$\operatorname*{argmax}_{s\in S}W_s$, which is therefore the smallest candidate
and always meets the standing condition, since $L_s=0$ gives $W_s=\infty$.  In
particular:
\begin{itemize}
\item if $\mathcal A=\{s^*\}$ is a singleton, then with probability tending to
one every surviving actor has signature $s^*$;
\item if the signatures in $\mathcal A$ are themselves disconnected, that is
if $\mathcal A$ admits a partition $\mathcal A=\mathcal A_1\sqcup\mathcal A_2$
into nonempty parts with $s\cap s'=\emptyset$ for every
$s\in\mathcal A_1,\ s'\in\mathcal A_2$, then with probability tending to one the
survivors occupy at least two components of $G_C$ at that time.
\end{itemize}
\end{lemma}

\begin{proof}
See Appendix~\ref{app:proofs-recovery}.
\end{proof}

\noindent
Divergence of the gap is what the statement needs, not mere strictness: with the
gap bounded, \eqref{eq:persistbound} does not force the probability to one
and the limit is nondegenerate.

\begin{corollary}[Load minimum as a special case]\label{cor:loadmin}
In the setting of Lemma~\ref{lem:persistmax} put $L^*:=\min_{s\in S}L_s$,
$\mathcal A:=\operatorname*{argmin}_{s\in S}L_s$ and
$\delta:=\min_{s\in S\setminus\mathcal A}(L_s-L^*)$, so $\delta\ge1$ because
each $L_s$ counts deficient axes.  Assume \textup{(H1)} the non-minimal classes
are bounded, i.e.\ there is $C$ with $n_s\le C$ for every
$s\in S\setminus\mathcal A$, and this bound does not grow with the population;
and \textup{(H2)} every minimal class diverges, i.e.\ $n_s\to\infty$ for every
$s\in\mathcal A$.  Then
$\min_{s\in\mathcal A}W_s-\max_{s\notin\mathcal A}W_s\to\infty$, so the
conclusions of Lemma~\ref{lem:persistmax} hold, and \eqref{eq:persistbound}
becomes, for every $x>0$,
\[
  \Pr\bigl[\mathrm{LAST}(\mathcal A)\bigr]\;\ge\;
  1-\bigl|S\setminus\mathcal A\bigr|\,C\,e^{-\lambda(L^*+\delta)x}
   -\sum_{s\in\mathcal A}\exp\bigl(-n_s\,e^{-\lambda L^*x}\bigr).
\]
\end{corollary}

\begin{proof}
See Appendix~\ref{app:proofs-recovery}.
\end{proof}

\noindent
The lemma and the corollary concern one quantity at the two scales of
Remark~\ref{rem:leadingorder}; the corollary gives a bound in terms of $C$,
$L^*$, $\delta$ and the sizes $n_s$ of the minimal-load classes.

\begin{remark}[Why the second case needs a partition]\label{rem:loadminsplit}
It is not enough that $\mathcal A$ merely \emph{contain} two disjoint signatures:
Lemma~\ref{lem:signatureintersection} forbids an edge between actors with those
signatures, not a path, and for $r\ge3$ a chain of further members of
$\mathcal A$ can join them.
With $r=3$, $\ell_1=\ell_2=1$, $\ell_3=0$, each requirement mask retaining at
least one nondeficient axis as in Theorem~\ref{thm:thinbridge}(iii), and
$S=\{\{1\},\{2\},\{1,3\},\{2,3\},\{1,2\}\}$, the separating set is everything but
$\{1,2\}$, and the disjoint signatures $\{1\}$ and $\{2\}$ are joined through
$\{1,3\}$ and $\{2,3\}$.
\end{remark}

\begin{proposition}[Simultaneous removal has positive probability]\label{prop:simultaneous}
Let $\beta>0$ and let $x$ be a state in which some axis $g$ has $|D_g|\ge2$.
Definition~\ref{def:law} gives finitely many exponential clocks, so the total
intensity $\Lambda(x)$ out of $x$ is finite, and
\[
  \Pr\bigl(\text{the $g$-common exit is the next transition}\mid X=x\bigr)
  \;=\;\frac{\beta\lambda}{\Lambda(x)}\;>\;0 .
\]
On that event every actor in $D_g$ is removed at once, so at least two actors
exit simultaneously.
\end{proposition}

\begin{proof}
See Appendix~\ref{app:proofs-recovery}.
\end{proof}

\noindent
For any $\beta\in(0,1]$ two actors deficient on a common axis are removed
together with positive probability (Proposition~\ref{prop:simultaneous}), and at
$\beta=1$ route selection also becomes determinate.  Actors removed by one
firing share the deficient axis $g$.

\begin{proposition}[Route selection at $\beta=1$]\label{prop:routeselector}
Consider the shared-state frozen slice of Definition~\ref{def:frozen} with
$\beta=1$ and $r=2$, in the $A/AB/B$ configuration of
Lemma~\ref{lem:uniquemediator}, with all three signatures present and
$\ell_A,\ell_B\ge1$.  For each deficient axis $g$ in mask $A$, $D_g(t)$ consists
of the actors with signatures $A$ and $AB$; for each deficient axis $g$ in mask
$B$, of the
classes $B$ and $AB$.

\begin{enumerate}
\item[(1)] At $\beta=1$ the idiosyncratic exit clocks of Definition~\ref{def:law}
have rate $(1-\beta)\lambda=0$ and never fire, and on the frozen slice the loss
and supply clocks are silent as well.  Only the axis-common exit clocks survive:
one per deficient axis, each of rate $\beta\lambda=\lambda$, and mutually
independent.

\item[(2)] Let $\tau$ be the time of the first axis-common exit among all
deficient axes.  Immediately afterward, at $\tau+$, the surviving population
consists of a single signature class: $B$ alone if the firing axis lies in
mask $A$, or $A$ alone if it lies in mask $B$.

\item[(3)] Let ``$A$-side first'' denote the event that the earliest-firing
deficient axis lies in mask $A$.  Then
\[
  \Pr(A\text{-side first})=\frac{\ell_A}{\ell_A+\ell_B} .
\]
\end{enumerate}
\end{proposition}

\begin{proof}
See Appendix~\ref{app:proofs-recovery}.
\end{proof}

\noindent
Positive loads at both endpoints are what make this a comparison: with no
deficient axis on one side there would be no competing clock there.

\begin{proposition}[Identical marginals, different outcomes]\label{prop:divergentendpoints}
Take the shared-state frozen slice of Definition~\ref{def:frozen} with $r=2$ in
the $A/AB/B$ configuration of Lemma~\ref{lem:uniquemediator}, all three
signatures present, and $\ell_A=\ell_B=\ell\ge1$.  Consider a sequence of
populations in which $S$ and the loads are held fixed and the mixed class is
thin in the persistence scale, as in Theorem~\ref{thm:thinbridge}(iv).  Let
$\tau$ be the instant at which the last actor in the mixed class exits.  Then
\begin{enumerate}
\item[(i)] at $\beta=0$, with probability tending to one the surviving signature
set at $\tau+$ is $\{A,B\}$, and the survivors occupy at least two components of
$G_C$;
\item[(ii)] at $\beta=1$, almost surely the surviving signature set at $\tau+$ is
a singleton, equal to $\{A\}$ or to $\{B\}$ with equal probability;
\item[(iii)] the marginal lifetime law of every actor is the same in (i) and
in (ii).
\end{enumerate}
\end{proposition}

\begin{proof}
See Appendix~\ref{app:proofs-recovery}.
\end{proof}

\noindent
At $\beta=1$ both selective attrition and correlated collapse hold and every
marginal lifetime law is the one at $\beta=0$, yet the generated outcome at
$\tau+$ differs, the difference produced by the axis-common channel.

\subsection{How $\beta$ changes population outcomes}

The parameter $\beta$ selects which route generates the morphologies of
Section~\ref{sec:intro}.  At $\beta=0$ the axis-common channel is
inactive and the frozen-slice exit clocks are independent, so, by
Proposition~\ref{prop:supporthazard}, actors with larger loads have
stochastically shorter lifetimes.
Selective attrition gives the direction of the attrition, the mixed class
over-represented among the exits
(Proposition~\ref{prop:enrichment}); Lemma~\ref{lem:persistmax} gives the
endpoint, the survivors being a set of signatures whose persistence scale
separates from the rest.  That is (M1) when that set is a singleton, and the shape of (M3) when its
signatures are disconnected among themselves;
Corollary~\ref{cor:loadmin} is the case in which that set is the set of least load.  The two cases are mutually exclusive: a singleton admits no partition into disjoint nonempty parts.  At $\beta=1$ the idiosyncratic clocks are silent and the first exit
event removes every actor deficient on one axis at once
(Proposition~\ref{prop:routeselector}), so correlated collapse yields exits that
coincide in time.  That is (M2).  Simultaneity is built into the axis-common
clock; what the results derive is who exits together, the actors deficient on the
firing axis, the route probability $\ell_A/(\ell_A+\ell_B)$
(Proposition~\ref{prop:routeselector}), and the invariance of every marginal
(Proposition~\ref{prop:divergentendpoints}).  Proposition~\ref{prop:nonidentity} is the minimal form of
the tension between the routes, Proposition~\ref{prop:divergentendpoints} the
sharp form.

(M3) comes from bridge depletion and fragmentation through
Theorem~\ref{thm:thinbridge}, and turns on the time at which the survivors are
inspected rather than on $\beta$: at $\beta=0$ the theorem describes them once
the bridge is exhausted and returns the fragmented shape, while
Lemma~\ref{lem:persistmax} describes them once everything outside the
persistence maximum is exhausted.  The fragmented shape is the earlier stage of one
trajectory.

In the two-mask slice $L_{AB}=\ell_A+\ell_B$ exceeds both endpoint loads, yet
when $n_A^{(\ell_A+\ell_B)/\ell_A}\ll n_{AB}$ and
$n_B^{(\ell_A+\ell_B)/\ell_B}\ll n_{AB}$, with $\ll$ as in
Remark~\ref{rem:thinness}, the separating set is $\{AB\}$ and every survivor
has the mixed signature, which is (M1) with the mediator as the surviving class.
Remark~\ref{rem:leadingorder} says why this agrees with the load ordering: the
load ordering applies to individual actors, whereas which class outlasts the
others is a statement about maxima over classes of different size.

\subsection{Latent generation and projected appearance}

Apparent cohesion sits apart: the other mechanisms describe the latent process;
apparent cohesion describes what a projection of that process makes visible.
A projection that discards upstream structure can present the survivors as
internally cohesive whether or not that cohesion is a feature of the latent
state.  Concentration in time is one instance: the two
routes derive the concentration, and the projections of
Section~\ref{sec:observables}, which take one snapshot of the cohesion graph
and forget which axis supported an adjacency, do not separate them.

%% file: sections/06_boundaries.tex
\section{Boundaries}\label{sec:boundaries}

This section says when forced coincidence holds, how often, what it implies for
an observable, and how the geometry weakens as $r$ grows, and where the dynamic results of Sections~\ref{sec:generation}--\ref{sec:recovery} stop once the coordinates they hold fixed are relaxed.  The characterization
and the observable statements are exact and measure-free; the proportion requires
a measure, stated where it is used.  The dynamics change which signature types
remain represented and which edges of $G_C(t)$ are realized, but not the
intersection rule defining $H_r$, so the answers apply to whatever realized family is present at the moment of evaluation.  Section~\ref{sec:coordinates} takes up what the dynamics do change.

\subsection{Realized families and forced coincidence}

\begin{definition}[Realized family]\label{def:skeleton}
The \emph{realized family} is the set $S\subseteq\Sigma_r$ of the
signatures present in a given population.  Write $G[S]:=H_r[S]$ for the
subgraph of the ambient graph induced on $S$.
\end{definition}

\noindent
Fix a realized family $S$ and disjoint $s,t\in S$, and put $u:=s\vee t=s\cup t$.  Say
that \emph{forced coincidence} holds at $(S,\{s,t\})$, written
$\mathrm{GOOD}(S,\{s,t\})$, when $s$ and $t$ are connected in $G[S]$ and $u$ is
the unique cut vertex separating them.  Call a set of vertices whose removal
leaves $s$ and $t$ in different components of $G[S]$ a \emph{separator} for the
pair; a cut vertex is a separator of size one.  If $\{u\}$ separates then so
does any superset, so uniqueness is asserted among single vertices.

\begin{lemma}[Forced coincidence is a union-cut condition]\label{lem:unioncut}
$\mathrm{GOOD}(S,\{s,t\})$ holds if and only if $u\in S$ and $s,t$ lie in
distinct components of $G[S\setminus\{u\}]$.
\end{lemma}

\begin{proof}
See Appendix~\ref{app:proofs-boundaries}.
\end{proof}

\begin{corollary}[Connectivity form]\label{cor:kappaform}
$\mathrm{GOOD}(S,\{s,t\})$ holds if and only if $u\in S$ and
$\kappa_{G[S]}(s,t)=1$, where $\kappa$ denotes the two-terminal local vertex
connectivity \citep{diestel2017}.
\end{corollary}

\begin{proof}
If $u\in S$ then $s-u-t$ is a path, so $s$ and $t$ are connected and
$\kappa\ge1$.  Suppose $\kappa=1$ and let $w$ be a separator.  If $w\neq u$ the
path $s-u-t$ survives its removal, a contradiction; so $w=u$ and
Lemma~\ref{lem:unioncut} applies.  Conversely
Lemma~\ref{lem:unioncut} gives a separator of size one.
\end{proof}

\begin{remark}[The characterization is not a graph invariant]\label{rem:notinvariant}
The condition refers to $u=s\vee t$, fixed by the requirement lattice rather than
read off the graph, and is not determined by the isomorphism type of
$(G[S],\{s,t\})$: the families $\{\{1\},\{1,2\},\{2\}\}$ and
$\{\{1\},\{1,2,3\},\{2\}\}$ both induce the path $s-\text{mid}-t$ on
$s=\{2\},t=\{1\}$, yet only the first has $u\in S$.  The graph supplies the
separator; the requirement structure supplies which vertex must serve as one.
\end{remark}

\subsection{Redundancy, and the shape of failure}

\begin{definition}[Common neighbours and redundancy]\label{def:bridgeset}
For disjoint $s,t$ in the realized family $S$ fixed above put
\[
  B(s,t):=\{\,v\in\Sigma_r\ :\ v\cap s\neq\emptyset\ \text{and}\ v\cap t\neq\emptyset\,\},
  \qquad
  R_{s,t}:=|S\cap B(s,t)\setminus\{s,t\}| .
\]
\end{definition}

\noindent
$B(s,t)$, the \emph{common-neighbour set} of the pair, is the set of signatures
meeting both endpoints; equivalently, the transversals of the two-edge
hypergraph $\{s,t\}$.  Its members are exactly the possible length-two
mediators, so $R_{s,t}$ counts the common neighbours of $s$ and $t$ in $G[S]$.
Only $B(s,t)$ is a function of the pair alone; $R_{s,t}$
depends on $S$ as well, and the subscript is abbreviated because $S$ is fixed
throughout.

\begin{lemma}[Two common neighbours preclude forcing]\label{lem:twobridges}
$R_{s,t}\le\kappa_{G[S]}(s,t)$.  In particular $R_{s,t}\ge2$ implies
$\neg\mathrm{GOOD}(S,\{s,t\})$.
\end{lemma}

\begin{proof}
Distinct $v_1,\dots,v_R\in S\cap B(s,t)$ give internally disjoint paths
$s-v_i-t$, so by Menger's theorem \citep{menger1927}, in vertex form
\citep[Theorem~8.10]{kortevygen2012},
$\kappa\ge R$.  If $R\ge2$ then $\kappa\ge2$
and Corollary~\ref{cor:kappaform} fails.
\end{proof}

\begin{lemma}[Size of the common-neighbour set]\label{lem:bridgesize}
If $|s|=a$ and $|t|=b$ with $s\cap t=\emptyset$ then
$|B(s,t)|=(2^{a}-1)(2^{b}-1)2^{\,r-a-b}\ \ge\ 2^{\,r-2}$, with equality at
$a=b=1$.
\end{lemma}

\begin{proof}
A member of $B(s,t)$ is determined by a nonempty trace on $s$, a nonempty trace
on $t$, and an arbitrary trace on the rest.  The product is minimized over
$a,b\ge1$ at $a=b=1$.
\end{proof}

\noindent
Combining Corollary~\ref{cor:kappaform} with Lemma~\ref{lem:twobridges} splits
failure into three cases:
\[
  \neg\mathrm{GOOD}
  \iff
  \underbrace{u\notin S}_{\text{join absent}}
  \ \ \text{or}\ \
  \underbrace{u\in S,\ \kappa\ge2}_{\text{redundant}} ,
\]
where the redundant case divides further according to whether the redundancy is
direct, $R_{s,t}\ge2$, or arises through a longer detour around $u$ with
$R_{s,t}\le1$.

\subsection{Counting measure and the asymptotic boundary}

Everything so far is a property of a given realized family and pair.  Asking how common
forced coincidence is requires a measure, and we state one.

\begin{definition}[Sample space and proportion]\label{def:omega}
Let $r\ge2$, so that disjoint nonempty signatures exist, and set
\[
  \Omega_r:=\bigl\{\,(S,\{s,t\})\ :\ S\subseteq\Sigma_r,\ H_r[S]\ \text{connected},\
  s,t\in S,\ s\cap t=\emptyset\,\bigr\},
\]
with $S$ labelled.  Throughout this subsection, and only here, proportions are
taken under the uniform counting measure on $\Omega_r$, and we write
\[
  \mathrm{FAIL}(r):=\Pr_{\omega\sim\mathrm{Unif}(\Omega_r)}\bigl[\neg\mathrm{GOOD}(\omega)\bigr].
\]
\end{definition}

\begin{theorem}[Forced coincidence is asymptotically negligible]\label{thm:failasymptotic}
Under the uniform counting measure on $\Omega_r$,
\[
  1-\mathrm{FAIL}(r)\ \le\ 2^{\,1-2^{\,r-2}}
  \ \longrightarrow\ 0
  \qquad (r\to\infty).
\]
\end{theorem}

\noindent
The bound is a statement about the uniform measure: conditioning on $|S|$ can
raise the proportion above it, and at $r=4$, $|S|=3$ the enumerated value is
$\pvStratForcedRFourSThree/\pvStratPairsRFourSThree$ against $1/8$
(Appendix~\ref{app:orbit}).

\begin{proof}
See Appendix~\ref{app:proofs-boundaries}.
\end{proof}

\noindent
The theorem is asymptotic, and the enumeration remains the finite statement.
The bound is exactly one at $r=2$, where the enumeration finds every admissible
pair forced, and it is informative from $r=3$ on, though it still stands above
the enumerated value at $r=3$ and $r=4$ (Appendix~\ref{app:orbit}).

\subsection{An observable boundary}

\begin{definition}[Morphology observables]\label{def:observables}
For a population $R$ with $|R|\ge2$ and cohesion graph $G_C$, let $M$ be the
number of
adjacent pairs, $Q$ the number of pairs lying in a common component, and $A$ the
number of pairs in a common component that share a required mask.  Put
\[
  \sigma:=\frac{M}{\binom{|R|}{2}},\qquad
  I:=\frac{Q}{\binom{|R|}{2}},\qquad
  K:=\frac{M}{Q},\qquad
  P_A:=\frac{A}{Q},
\]
where $K$ and $P_A$ are defined when $Q>0$, and then $\sigma=I\cdot K$.  We call $\sigma$ the unconditional edge density, $K$
the edge density conditional on reachability, $I$ the integration observable,
the probability that a uniformly drawn pair is mutually reachable, and $P_A$ the
pair cover.
\end{definition}

\begin{proposition}[Edge-density cap under forced coincidence]\label{prop:forced-edge-density-cap}
Let $R$ be a population whose realized family is $S$ with $|S|\ge3$, containing
$n_v\ge1$ actors of signature $v$ for each $v\in S$, and put $N:=|R|=\sum_{v\in S}n_v$.
Suppose the disjoint pair $\{s,t\}\subseteq S$ is forced and write
$u:=s\vee t$.  Then
\[
  \sigma\;\le\;1-
  \frac{\min\bigl\{\,n_s\,(N-n_u-n_s),\ \,n_t\,(N-n_u-n_t)\,\bigr\}}
       {\binom N2}.
\]
With one actor per signature, so that $N=|S|=:m$, the right-hand side is
\[
  \mathrm{cap}(m)
  \;:=\; \frac{\binom{m-1}{2}+1}{\binom m2}
  \;=\; 1-\frac{2(m-2)}{m(m-1)} .
\]
\end{proposition}

\begin{proof}
See Appendix~\ref{app:proofs-boundaries}.
\end{proof}

\noindent
No connectivity assumption is needed, which is what makes the bound usable: it
holds whatever the realized graph looks like.

\begin{remark}[The bound is a function of the composition]\label{rem:capcomposition}
The cap is not a number attached to the realized family.  Holding $S$ and the forced
pair fixed and letting $n_u$ grow sends the right-hand side to $1$, so the
bound becomes arbitrarily weak as the mediating class grows.  An observer who
reads $\sigma$ off the graph and does not
know $n_s$, $n_t$ and $n_u$ therefore cannot use the bound to exclude anything,
and Corollary~\ref{cor:inverseexclusion} states the exclusion where the
composition is known.  This is one
instance of the reading of Section~\ref{sec:observables}: what discriminates
lies upstream of the graph.
\end{remark}

\begin{remark}[The cap is on $\sigma$, not on $K$]\label{rem:whichobservable}
The bound is on the unconditional edge density.  Since $K=M/Q$, fragmentation
reduces $Q$ and can raise the conditional edge density even while $\sigma$ stays
below the cap.  Take $r=4$ and
$S=\{\{1\},\{2\},\{3\},\{2,3\},\{1,4\}\}$ at unit composition, forced at
$s=\{2\}$, $t=\{3\}$: here
$\sigma=3/10$ lies below $\mathrm{cap}(5)=7/10$ while $K=3/4$ exceeds it.  The
bound does not apply to $P_A$ either, since $P_A$ counts pairs sharing a
required mask whether or not that mask is currently satisfied.  When the
realized graph is connected and every required
axis is satisfied, $\sigma=K=P_A$ and $I=1$.
\end{remark}

\subsection{Graded dependence on \texorpdfstring{$r$}{r}}

Finally, the geometry depends quantitatively on $r$ at every support
size, not only at the boundary between $r=2$ and $r\ge3$.

\begin{proposition}[Cover bounds]\label{prop:coverbounds}
Let $C$ be a component of $G_C$ of size $n\ge2$ whose signatures lie in $\Sigma_r$, and
write $P_A(C)$ for the pair cover computed within $C$.  Then
$P_A(C)\ge 1-t_r(n)/\binom n2$, where $t_r(n)$ is the number of edges of the
$r$-partite Tur\'an graph; in particular $P_A(C)\ge 1/r-O(n^{-1})$.  Moreover
$P_A(C)=1$ if and only if the signatures in $C$ are pairwise intersecting.
\end{proposition}

\begin{proof}
See Appendix~\ref{app:proofs-boundaries}.
\end{proof}

\noindent
The floor $1/r$ falls as $r$ grows, so the geometry weakens gradually rather
than switching at a threshold.

\input{sections/06b_coordinates}

%% file: sections/06b_coordinates.tex
\subsection{The boundary in the dynamic coordinates}\label{sec:coordinates}

Section~\ref{sec:dynamics} declared four free coordinates and varied only
$\beta$; the table summarizes the effects of the other three.  A conclusion
\emph{holds} when it stands verbatim off the frozen slice of
Definition~\ref{def:frozen}.  It \emph{changes} when the ordering or mechanism
stands and the deciding quantity is replaced.  It \emph{fails} when the
qualitative content is lost on some configuration.  Throughout this section
$g\in s_i$ abbreviates $g\in\bigcup_{m\in s_i}\mathcal G_m$, the axes the masks
of $i$ require.  Write $\xi_i(t):=\mathbf1\{i\in R(t)\}$,
$\varphi_i(t):=(\varphi_{ig}(t))_{g\in G}$ and
$E_\cap(t):=\{\{i,j\}\subseteq R(t):i\neq j,\ s_i\cap s_j\neq\emptyset\}$, the
intersection graph on survivors, which is independent of the satisfaction state
and contains $E_C(t)$ (Lemma~\ref{lem:signatureintersection}).

\begin{center}
\small
\begin{tabular}{@{}p{0.18\linewidth}p{0.25\linewidth}p{0.15\linewidth}p{0.26\linewidth}@{}}
\toprule
 & $\nu$ & $\rho_S$ & $\rho$ \\
\midrule
Marginal invariance
  & holds (Prop.~\ref{prop:coord-lawinvariance})
  & holds (ibid.) & holds (ibid.) \\
Structural layer
  & holds; rate-independent & holds; rate-independent & holds; rate-independent \\
Load and order (Prop.~\ref{prop:supporthazard})
  & fails: independence is a property of $\nu=0$ (Prop.~\ref{prop:coord-lossdependence}), and inclusion no longer orders (Prop.~\ref{prop:coord-loss}(ii))
  & holds; degenerate (Rem.~\ref{rem:gauge})
  & changes: defective lifetime, ordering kept (Prop.~\ref{prop:coord-supply}) \\
Structural boundary (\S\ref{sec:boundaries})
  & holds; rate-independent & holds; rate-independent & holds; rate-independent \\
Attrition, (M1) (Lem.~\ref{lem:persistmax})
  & changes: effective load, last class can reverse (Prop.~\ref{prop:coord-loss})
  & holds; degenerate
  & changes: permanent survival replaces $W_s$ (Prop.~\ref{prop:coord-supply}) \\
Fragmentation, (M3) (Thm.~\ref{thm:thinbridge})
  & fails at $\rho_S>0$; changes at $\rho_S=0$ (Prop.~\ref{prop:coord-loss})
  & holds; degenerate
  & changes: closed form at finite $t$, no limit one (Prop.~\ref{prop:coord-supply}) \\
Concentration, (M2) (Prop.~\ref{prop:routeselector})
  & changes: a second route at $\beta=0$ (Prop.~\ref{prop:coord-lossdependence})
  & changes: $\beta$ sufficient, not necessary (ibid.)
  & changes: only Prop.~\ref{prop:simultaneous} is unrestricted \\
Support of exit dependence
  & holds: only the sign of $\rho_S\nu$ matters (ibid.)
  & holds: $\rho_S=0$ gives independence
  & changes: $E_C(0)$ becomes $E_\cap$ (Rem.~\ref{rem:coord-support}) \\
\bottomrule
\end{tabular}
\end{center}

\begin{proposition}[Marginal invariance at the level of laws]\label{prop:coord-lawinvariance}
Fix $\lambda>0$, $\nu,\rho\ge0$ and a deterministic initial state
$(R_0,\varphi(0))$.  For each $i\in P$ the law of
$\bigl(\varphi_i(t),\xi_i(t)\bigr)_{t\ge0}$ is the same for every
$(\rho_S,\beta)\in[0,1]^2$.
\end{proposition}

\begin{proof}
See Appendix~\ref{app:proofs-coordinates}.
\end{proof}

Lemma~\ref{lem:marginalinvariance} fixes rates, not laws; the proof here uses
that the rates of Definition~\ref{def:law} are constants.

\begin{proposition}[\texorpdfstring{$\beta=0$}{beta=0} alone does not make exit times independent]\label{prop:coord-lossdependence}
Fix $\rho=0$, $\beta=0$, $\nu>0$ and a deterministic $\varphi(0)$.  For $i\in P$
put $A_i:=\{\,g\in s_i:\varphi_{ig}(0)=1\,\}$, and for a pair $i\neq j$ put
$C:=A_i\cap A_j$ and
$D(s,t):=\Pr(T_i\le s,T_j\le t)-\Pr(T_i\le s)\Pr(T_j\le t)$.  If $C=\emptyset$
or $\rho_S=0$, then $T_i$ and $T_j$ are independent.  If $C\neq\emptyset$ and
$\rho_S>0$, then $D(s,t)>0$ for every $s,t>0$.
\end{proposition}

\begin{proof}
See Appendix~\ref{app:proofs-coordinates}.
\end{proof}

So $\beta>0$ is sufficient and not necessary for the coupling in
Proposition~\ref{prop:supporthazard}: at $\beta=0$ a shared satisfied axis with
$\rho_S>0$ already couples exit times.  The coordinate carrying temporal
concentration remains $\beta$, sharp at the endpoint
(Proposition~\ref{prop:routeselector}) on the slice, where $\rho_S$ degenerates
(Remark~\ref{rem:gauge}); the loss-side route needs $\nu>0$ and is a second
sufficient condition, not a substitute.

\begin{proposition}[Supply]\label{prop:coord-supply}
Fix $\nu=0$, $\beta=0$ and $\rho>0$, so that $\rho_S$ degenerates
(Remark~\ref{rem:gauge}), and put $q:=\rho/(\lambda+\rho)$ and
$h(t):=\bigl(\rho+\lambda e^{-(\lambda+\rho)t}\bigr)/(\lambda+\rho)$.
(i) An actor with $L$ deficient incidences at time zero has
$\Pr(T>t)=h(t)^{L}$ and permanent survival probability $\Pr(T=\infty)=q^{L}$,
both strictly decreasing in $L$, and the lifetimes of distinct actors are
mutually independent because $\nu=0$ and $\beta=0$, not because of the
frozen slice.
(ii) Class $s$ dies out with probability $(1-q^{L_s})^{n_s}$, so no class dies
out in the limit $n_s\to\infty$; the expected number of permanent survivors
of $s$ is $n_sq^{L_s}$, so as every $n_s\to\infty$ their composition is
proportional to $n_sq^{L_s}$, in which the mixed class is under-represented
against pure $A$ by the factor $q^{\,L_{AB}-L_A}$; and, with $S$, the loads and
$(n_s)_{s\notin\mathcal A}$ fixed and $n_s\to\infty$ for $s\in\mathcal A$,
$\Pr[\mathrm{LAST}(\mathcal A)]\to\prod_{s\notin\mathcal A}(1-q^{L_s})^{n_s}$.
(iii) The conclusion of Proposition~\ref{prop:enrichment} holds at every
reference time $\tau\ge0$ and horizon $t>0$; with $\ell_A=\ell_B=1$ and the three
classes equally numerous, $E_{AB}$ increases in $\tau$ to $(2+q)/(1+q)$ at every
horizon.
(iv) With $\ell_A=\ell_B=\ell$, the probability that at time $t$ the mixed class
is extinct and both pure classes present is
\[
  P_{\mathrm{frag}}(t)=\bigl(1-h(t)^{2\ell}\bigr)^{n_{AB}}
    \bigl(1-(1-h(t)^{\ell})^{n_A}\bigr)\bigl(1-(1-h(t)^{\ell})^{n_B}\bigr),
\]
and at fixed $\ell$, $\sup_tP_{\mathrm{frag}}(t)\le1-q^{2\ell}<1$.
(v) Only exit removes edges: $E_C(t_1)\cap\binom{R(t_2)}{2}\subseteq E_C(t_2)$
for $t_1\le t_2$, and on the permanent survivors $E_C$ is nondecreasing and
increases to the restriction of $E_\cap$ to them.
\end{proposition}

\begin{proof}
See Appendix~\ref{app:proofs-coordinates}.
\end{proof}

Classes are ordered by $q^{L_s}$: (M1) keeps its ordering, while
$\log n_s/(\lambda L_s)$ is no longer the appropriate scale.  Part (iv) is the
finite-$t$ analogue of Theorem~\ref{thm:thinbridge}, but its supremum stays below one at fixed $\ell$, so the limit statement of that theorem has no analogue unless $\ell\to\infty$, which its hypothesis (ii) forbids; and the double limit does not commute: $\rho\to0$ at fixed multiplicities
recovers that proof, $n\to\infty$ at fixed $\rho>0$ does not.  Part (v) yields
the internal cohesion that hypothesis (iii) of that theorem supplies, but only
in the limit and only among the permanent survivors; a clique needs $P_A=1$
in addition (Proposition~\ref{prop:coverbounds}).

\begin{proposition}[Loss]\label{prop:coord-loss}
Fix $\rho=0$, $\beta=0$, $\nu>0$ and $\rho_S\in[0,1]$.  For an actor $i$ put
$k_i:=\sum_{m\in s_i}|\mathcal G_m|$, the number of axes $i$ requires, and
$L_i:=d_i(0)$, write $\widehat\nu:=\min(\lambda,\nu)$, and
define the \emph{effective load}
$L^{\mathrm{eff}}_i:=L_i+(k_i-L_i)\widehat\nu/\lambda\in[L_i,k_i]$; for a class
$s$ whose actors agree on both counts write $k_s,L_s,L^{\mathrm{eff}}_s$.
(i) $\Pr(T_i>t)=e^{-\lambda L_it}B(t)^{\,k_i-L_i}$, where
$B(t)=(\lambda e^{-\nu t}-\nu e^{-\lambda t})/(\lambda-\nu)$, equal to
$e^{-\lambda t}(1+\lambda t)$ at $\nu=\lambda$; the marginal does not depend on
$\rho_S$, a case of Proposition~\ref{prop:coord-lawinvariance}, and
$-t^{-1}\log\Pr(T_i>t)\to\lambda L^{\mathrm{eff}}_i$.
(ii) If $L_j\ge L_i$ and $k_j-L_j\ge k_i-L_i$ then $T_j\le_{\mathrm{st}}T_i$;
signature inclusion alone does not give the ordering.
(iii) With $M_s$ the time of the last exit in a class of $n_s$ actors,
$M_s=\log n_s/(\lambda L^{\mathrm{eff}}_s)+O_P(1)$ at $\rho_S=0$ and
$\nu\neq\lambda$, with a further $O_P(\log\log n_s)$ at $\nu=\lambda$; and
$M_s=\log n_s/(\lambda k_s)+O_P(1)$ at $\rho_S>0$.
(iv) For $\nu\neq\lambda$, Lemma~\ref{lem:persistmax} survives with $L_s$
replaced by $L^{\mathrm{eff}}_s$ outside $\mathcal A$, at the cost of a constant
$C_s$ multiplying the first sum, and by $k_s$ inside: for every $x>0$,
\[
  \Pr[\mathrm{LAST}(\mathcal A)]\ \ge\
    1-\sum_{s\notin\mathcal A}n_sC_se^{-\lambda L^{\mathrm{eff}}_sx}
     -\sum_{s\in\mathcal A}\exp\bigl(-n_se^{-\lambda k_sx}\bigr),
\]
with separation condition $\min_{\mathcal A}\log n_s/(\lambda k_s)
-\max_{\mathcal A^{c}}\log n_s/(\lambda L^{\mathrm{eff}}_s)\to\infty$.  Two
scales meet: the condition measures $\mathcal A$ in $\log n_s/(\lambda k_s)$, the
scale of the class maximum under shared loss, and its complement in
$\log n_s/(\lambda L^{\mathrm{eff}}_s)$, the scale the marginal tails give; the
two coincide at $\nu>\lambda$, where $L^{\mathrm{eff}}_s=k_s$; at $\nu=\lambda$ the
outer tail carries the factor $(1+\lambda x)^{k_s-L_s}$, a shift of order
$\log\log n_s$ outside the condition.  The last surviving class can differ
from the one selected on the slice; the convention $W_s:=\infty$ at
$L_s=0$ does not survive, a class deficient on no axis at time zero being now
mortal.
(v) In the construction used in the proof, $\Theta':=\max\min(\sigma_g,
\varepsilon_{ig})$ over the incidences $(i,g)$ with $g\in s_i$ and
$\varphi_{ig}(0)=1$ is almost surely finite, and $E_C(t)=\emptyset$ for every
$t\ge\Theta'$.
(vi) If $\rho_S>0$ then $\Theta'\le\max_g\sigma_g$, which does not depend on the
multiplicities, while by (iii) the last mixed exit diverges as
$n_{AB}\to\infty$; hence along a sequence satisfying the separation condition of
(iv) at $\mathcal A=\{A,B\}$ and containing at least two pure-$A$ survivors at that
time, the probability that the surviving pure classes are internally cohesive
when the bridge is exhausted tends to zero.
\end{proposition}

\begin{proof}
See Appendix~\ref{app:proofs-coordinates}.
\end{proof}

Lifetimes are ordered by the effective load, which assigns weight $1$ to each
initially deficient axis and weight $\widehat\nu/\lambda$ to each initially
satisfied axis, and can disagree with $L_s$.  (M3) is no
longer a standing shape: the cohesion graph empties and fragmentation is visible only in a window, and (vi) is a counterexample to the
conclusion of Theorem~\ref{thm:thinbridge} outside its hypotheses, not to the
theorem.  At $\rho_S=0$ the loss clocks are unshared, a pure class atomizes
later the larger it is, and the cohesion clause survives under a second thin
condition against that growth.

\begin{remark}[Support of exit dependence under supply]\label{rem:coord-support}
At $\beta=0$, $\nu>0$ and $\rho_S>0$, the regime in which
Proposition~\ref{prop:coord-lossdependence} gives dependence, which pairs can have dependent
exit times is settled by which pairs can come to share a satisfied required
axis.  At $\rho=0$ satisfaction only decreases, so
that set is $E_C(0)$, and Proposition~\ref{prop:coord-lossdependence} identifies
it in both directions.  At $\rho>0$ a pair can acquire a common satisfied axis.  Enumeration over
$\pvCoordSupportEnum$ configurations (two actors, $r\le4$, every nonempty
signature and every initial satisfaction state) finds a state in which both are
alive and share a satisfied required axis reachable exactly when
$s_i\cap s_j\neq\emptyset$, so the reachable set is $E_\cap$; no proof is
offered.  The enumeration settles the reachable set, and the dependence on it is
Proposition~\ref{prop:coord-lossdependence}, stated at $\rho=0$.
\end{remark}

Propositions~\ref{prop:coord-lawinvariance} and \ref{prop:coord-lossdependence}
hold in a class wider than Definition~\ref{def:law}: rates
$\nu_{ig},\rho_{ig},\lambda_{ig}$ on incidences, and rates
$a_g\le\min_i\nu_{ig}$ and $b_g\le\min_i\lambda_{ig}$ on axes for the shared
loss and the axis-common exit.  Definition~\ref{def:law} is the specialization
$a_g=\rho_S\nu$, $b_g=\beta\lambda$ at constant rates.  One statement does not
transfer: the ordering part of Proposition~\ref{prop:supporthazard} needs
$\lambda_{ig}=\lambda_g$ and the axis-sharing of Definition~\ref{def:frozen}, both supplied by Definition~\ref{def:law}; dropping the latter produces a reversal even when the former holds.

The boundary drawn earlier in this section is a boundary in $r$: the three
roles align by force only at $r=2$, and under the counting measure of
Definition~\ref{def:omega} the proportion of family--pair instances
$(S,\{s,t\})$ in which the join separates the endpoints vanishes as masks
multiply.  The boundary
here is a boundary in the coordinates: off the slice each morphology is either
replaced by a named substitute (the permanent-survival fraction for the
persistence scale, the effective load for the load) or lost on a named
configuration.

%% file: sections/07_observables.tex
\section{Observable implications}\label{sec:observables}

This section asks what a measurement of the latent process can show: what the
theory predicts about observations, and what an observation rules out about the
latent structure.

\subsection{Projection and exact readouts}

Two steps separate the latent state from the observables.  From
$(R(t),\varphi(t))$ and the requirement structure one forms the cohesion graph
$G_C(t)$ of Definition~\ref{def:cohesion}; from that graph, $\sigma$, $I$ and
$K$; $P_A$ is computed from the graph together with the requirement structure,
since it counts pairs in a common component that share a required mask,
satisfied or not.
Each step discards information: the projection forgets which axis supported an
adjacency, and the numbers forget the graph.  Non-transitivity
can enter at the first step: by Remark~\ref{rem:nontransitivity} each per-axis
relation is transitive on its satisfaction support, so a failure of transitivity
in $G_C(t)$ is produced by the disjunction over axes and need not correspond to
anything intransitive upstream.

\begin{proposition}[Complete pair cover implies clique components]\label{prop:endpoint}
Suppose that every axis required by at least two surviving actors is
satisfied at every surviving incidence requiring it, that is
$\varphi_{ig}=1$ for all such $g$ and all surviving $i$ with $g\in s_i$, and
suppose the surviving requirement system has complete pair cover, $P_A=1$.  Then every component of $G_C$ is a clique, so $G_C$ is a
disjoint union of complete graphs; consequently $M=Q$ and $K=1$.
\end{proposition}

\begin{proof}
Let $i,j$ lie in the same component.  By $P_A=1$ they share a required mask, and
by hypothesis every axis of that mask is satisfied for both, so
Definition~\ref{def:cohesion} makes them adjacent.  Hence each component is
complete, $M=Q$, and $K=M/Q=1$.
\end{proof}

\noindent
Proposition~\ref{prop:endpoint} is the formal content of apparent cohesion:
when pair cover is complete and every axis still required is satisfied by
everyone who requires it, each component appears perfectly cohesive,
whatever the history that produced it.

The second result runs the other way: for a two-component state, $I$ is
determined by the relative imbalance, at the scale $\sqrt{S_N}$.

\begin{proposition}[Balance window]\label{prop:balancewindow}
Consider a sequence of two-component separation states of the type produced by
Theorem~\ref{thm:thinbridge}, with pure survivor counts $n_A^{(N)}$ and
$n_B^{(N)}$ and total $S_N=n_A^{(N)}+n_B^{(N)}$.  Assume $S_N\to\infty$ in
probability and put
\[
  Z_N:=\frac{n_A^{(N)}-n_B^{(N)}}{\sqrt{S_N}} .
\]
If $Z_N\Rightarrow Z$ with $\Pr(|Z|=1)=0$, then
$\Pr(I_N<\tfrac12)\to\Pr(|Z|<1)$.  Hence any nondegenerate limiting imbalance
law with $0<\Pr(|Z|<1)<1$ gives a nondegenerate limiting probability that the
integration observable falls below one half.
\end{proposition}

\begin{proof}
See Appendix~\ref{app:proofs-observables}.
\end{proof}

\begin{corollary}[A model-generated balance law]\label{cor:modelbalance}
Specialize Theorem~\ref{thm:thinbridge} to equal endpoint loads
$\ell_A=\ell_B=\ell>0$, equal initial pure-class sizes $N$, and a single actor
of the mixed class, and write $\alpha:=\lambda\ell$.  Then, with $Z_N$ and $I_N$ as in
Proposition~\ref{prop:balancewindow},
\[
  Z_N\;\Longrightarrow\;Z:=\sqrt{1-\theta}\;G,
  \qquad \theta\sim\mathrm{Beta}(2,1),\quad G\sim N(0,1),
  \quad \theta\perp G,
\]
and $\Pr(|Z|=1)=0$, so Proposition~\ref{prop:balancewindow} applies and
\[
  \Pr\bigl(I_N<\tfrac12\bigr)\;\longrightarrow\;\Pr(|Z|<1)
  \;=\;\mathbb{E}\Bigl[\,2\Phi\bigl((1-\theta)^{-1/2}\bigr)-1\,\Bigr]
  \;\in\;\bigl(2\Phi(1)-1,\;1\bigr),
\]
where $\Phi$ is the standard normal distribution function.  The expectation
evaluates to $\pvBalanceReadout$.
\end{corollary}

\begin{proof}
See Appendix~\ref{app:proofs-observables}.
\end{proof}

\noindent
The process supplies its own imbalance law, and under it the observable verdict
is neither forced nor excluded: a separated state reads out as poorly integrated,
$I_N<\tfrac12$, in a limiting fraction $\pvBalanceReadout$ of realizations.  The
projection can present a purified population as cohesive, and a fragmented one
as integrated or not according to a quantity the fragmentation does not fix.

\subsection{A comparative prediction}

The forward direction has two tiers.  The structural tier is exact: by
Lemma~\ref{lem:twobridges}, if two or more signatures in the realized family meet both
endpoints then no single class separates them, so the coincidence of
Section~\ref{sec:generation} is not structurally forced, and that is a theorem.
The empirical tier does not follow from the structural tier without (A1)--(A4).

\begin{description}
  \item[(A1)] Diversity of requirement patterns corresponds to $r$.
  \item[(A2)] The presence of two or more distinct mediators meeting both
    endpoints corresponds to $R_{s,t}\ge2$, the focal moderator.  It counts
    mediators, not routes: $R_{s,t}\le\kappa_{G[S]}(s,t)$ can be strict, since a
    detour raises local connectivity without supplying a second length-two
    mediator.
  \item[(A3a)] Brokerage is measured by a pair-specific mediation observable for
    the endpoints in question, of which a separator is the extreme case.
  \item[(A3b)] Vulnerability is measured by an exposure observable tracking
    $L_s$, within the scope of Proposition~\ref{prop:supporthazard}; off the frozen slice the exposure quantity is not $L_s$ (Section~\ref{sec:coordinates}).
  \item[(A4)] Where structural forcing is absent, the distribution over realized
    configurations assigns positive probability, other things equal, to
    configurations outside the coincidence-preserving subset.
\end{description}

\noindent
Forced coincidence in the sense of Section~\ref{sec:boundaries} says that
$u=s\vee t$ is the unique cut vertex, not that it carries the largest load, and
at general $r$ the two come apart even when every $\ell_m$ is positive.  The
comparator is $r=2$ with $\ell_A,\ell_B>0$, where the three roles align
(Corollary~\ref{cor:forcedcoincidence}): $L_{AB}=L_A+L_B$ exceeds both, so
the unique mediator is also the unique load maximizer.

Write $F$ for that regime and $\mathcal{R}$ for a redundant-mediation regime in
the sense of (A2).  Under (A1)--(A4), with outcome
$Y=\mathbf{1}\{\text{brokerage class}=\text{vulnerability class}\}$, the theory
predicts
\[
  \Pr\bigl(Y=1\mid\mathcal{R}\bigr)
  \;<\;
  \Pr\bigl(Y=1\mid F\bigr) = 1 .
\]
The right-hand equality holds inside the model; (A4) supplies the strict
inequality, by denying that realized configurations under $\mathcal{R}$
concentrate on the coincidence-preserving subset.

\subsection{Inverse structural information}

The other direction asks what an observation rules out; the inverse
implications of the four projections differ in evidential status.

\begin{center}
\begin{tabular}{@{}lp{0.28\linewidth}p{0.30\linewidth}l@{}}
\toprule
Projection & Inverse information & Scope & Status \\
\midrule
$\sigma$ & $\sigma>\mathrm{cap}(m)\Rightarrow$ no disjoint pair is forced &
  at one actor per signature, where it is sharp
  (Appendix~\ref{app:extremal}); away from it the bound needs the composition
  (Remark~\ref{rem:capcomposition}) & \emph{Exact} \\
$K$ & partial support separation & observed in the audited domain &
  \emph{Enumerated} \\
$P_A$ & partial exclusion on the low side & observed in the audited domain &
  \emph{Enumerated} \\
$I$ & no one-dimensional threshold detected & in the audited domain &
  \emph{Enumerated} \\
\bottomrule
\end{tabular}
\end{center}

\begin{corollary}[Observable exclusion]\label{cor:inverseexclusion}
Let $R$ be a population as in
Proposition~\ref{prop:forced-edge-density-cap} with one actor per
signature, so that $|R|=|S|=m\ge3$.  If
$\sigma>1-2(m-2)/\bigl(m(m-1)\bigr)$ then no disjoint pair in the realized
family is forced.
\end{corollary}

\begin{proof}
Contrapositive of Proposition~\ref{prop:forced-edge-density-cap}.
\end{proof}

\noindent
Corollary~\ref{cor:inverseexclusion} excludes a structural configuration without
identifying the realized family or which mechanism of Section~\ref{sec:recovery}
operated; the other three projections are enumerated over the audited domain of
Appendix~\ref{app:supports}.

The choice of projection determines which upstream distinctions an observation
still constrains.  That is the sense in which the dispersion of descriptions and
the recurrence of morphologies are two aspects of one situation: different
readings of the same latent process keep different parts of it.

%% file: sections/08_discussion.tex
\section{Discussion}\label{sec:discussion}

\subsection{What is unified}

The unification is of the substrate, not of the mechanisms.  At two masks the
substrate forces mediation, closed-neighbourhood size and load onto one class,
and the independent derivations of (M3) and (M1) both concern that class:
unification without identification, exact at two masks and characterized at
every number of masks.

\subsection{Relation to existing accounts}

Each account named in Section~\ref{sec:intro} names a mechanism, and the model
recovers a regime with the corresponding outcome.  \citet{krackhardt1986snowball}
related clustered turnover to regular equivalence of communication roles, and
threshold accounts \citep{granovetter1978threshold} count the others who have
already gone; correlated collapse gives clustered removal from a
condition-level exit channel with no interpersonal transmission.  The
attraction--selection--attrition framework \citep{schneider1987asa}
corresponds to selective attrition, but there attrition
follows a perceived misfit, whereas here the mixed class is shortest lived
because it carries the summed load of its endpoints.  Structural balance
\citep{cartwright1956structural} corresponds to bridge depletion and
fragmentation and is driven by the tension of an unbalanced triad; neither
bridge depletion nor fragmentation here requires an actor to evaluate a
triad.
Intermediate-linkage models of withdrawal \citep{mobley1977intermediate} concern
the individual exit decision; the selection among the exits is selective
attrition.
In each case the morphology is derived without the behavioural content of the
account.  The static theory takes the same posture at a fixed positioning, where
a non-degenerate axis forces fragmentation and local cohesion together
\citep{imptheom}; that conjunction is (M3) here, and the dynamic extension adds
(M1) and (M2), which need the passage of time to state.

\subsection{Relation to neighbouring mathematical structures}

Appendix~\ref{app:comparisons} compares the characterization of
Section~\ref{sec:boundaries} with intersection graphs of lattice and ring
ideals, quorum systems, and simple games.  A separator statement in
graph theory takes the graph and a pair as given and asks whether some vertex
separates them; here the join $s\vee t$ is designated by the requirement
structure before the graph is consulted.  We know of no prior result that
selects the lattice join in advance for a designated pair in a
subset-intersection graph, characterizes when it is the unique cut vertex, and
ties that role to an independently derived load ordering.
Since that characterization depends on which classes are counted as present,
an indispensability observed in practice may be induced by the representation
rather than by the latent structure: restricting the actors, the conditions, or
the admissible paths removes alternative classes from $S$ and can make the
designated join the unique separator.

%% file: sections/09_conclusion.tex
\section{Conclusion}\label{sec:conclusion}

A finite condition structure with a minimal dynamic extension produces three
morphologies, each in a regime of its own, with no actor responding to another.
At two masks the requirement lattice forces the mediating, the best-connected
and the most exposed class to be one class, so the derivations of fragmentation
and homogenization, though independent, apply to a single position.  With
more than two masks the alignment is no longer forced.  Forced coincidence
in the sense of Section~\ref{sec:boundaries} admits an exact
characterization for every number of masks: it holds if and only if the
requirement-derived join lies in the realized family and separates the
endpoints.  The proportion of family--pair instances meeting this condition
becomes negligible under the stated counting measure as the number of masks
grows.  Which of the morphologies the literatures came to report is a question
about selection, and a different account answers it.

The load results hold on the shared-state frozen slice, and
Section~\ref{sec:coordinates} states what becomes of them off it.  The
comparative prediction rests on assumptions (A1)--(A4) of
Section~\ref{sec:observables}, stated there and not tested; testing them needs
measurements upstream of the graph.

%% file: sections/A_appendix.tex
\section{Proofs}\label{app:proofs}

The statements proved here are those of the body, and nothing is added to them.
Proofs are headed by the result each establishes and follow the order in which
results appear.

\subsection{Proofs from Section~\ref{sec:generation}}\label{app:proofs-generation}

\begin{proof}[Proof of Proposition~\ref{prop:supporthazard}]
An actor with signature $s$ has $L_s$ deficient incidences, and on the slice this set does
not change.  By Lemma~\ref{lem:marginalinvariance} each contributes total rate
$\lambda$ to the exit of that actor, so the exit time is exponential with rate
$\lambda L_s$.  Splitting the sum by mask gives the minimum representation.
Monotonicity of $L_s$ in $s$ is immediate from $\ell_m\ge0$, and a larger
exponential rate gives a stochastically smaller lifetime.  Independence at
$\beta=0$ is the disjointness of the clock index sets.

For $\beta>0$ the surviving exit clocks are the idiosyncratic ones, indexed by
deficient incidences and still disjoint across actors, together with one
axis-common clock per deficient axis, shared by every actor deficient on that
axis.  Writing $E_{ig}$ and $F_g$ for their independent exponential firing
times, the lifetime of an actor $i$ is $T_i=\min_g\min\{E_{ig},F_g\}$ over the
axes on which $i$ is deficient, so $(T_i)_i$ has the multivariate exponential
distribution generated by shared shocks.  Each $T_i$ is a nondecreasing function
of the same family of independent variables $(E_{ig},F_g)$, so the vector is
associated, and association gives
$\Pr(T_i>x,\,T_j>y)\ge\Pr(T_i>x)\Pr(T_j>y)$ for all $x,y$, which is the
asserted positive dependence.  It is strict whenever $i$ and $j$ are deficient
on a common axis $g$, since that axis then contributes
$e^{-\beta\lambda(x\vee y)}$ to the joint survival function in place of the
$e^{-\beta\lambda(x+y)}$ that independence would give.
\end{proof}

\subsection{Proofs from Section~\ref{sec:recovery}}\label{app:proofs-recovery}

\begin{proof}[Proof of Theorem~\ref{thm:thinbridge}]
Apply Lemma~\ref{lem:persistmax} on the slice with $S=\{A,AB,B\}$ and
$\mathcal A=\{A,B\}$.  The only signature outside $\mathcal A$ is $AB$, and
$L_{AB}=\ell_A+\ell_B>0$ by (i)--(ii), so the standing condition holds; the
separation hypothesis $\min_{s\in\mathcal A}W_s-\max_{s\notin\mathcal A}W_s
\to\infty$ is (iv).  Hence $\Pr[\mathrm{LAST}(\{A,B\})]\to1$: every mixed
actor has exited while a pure-$A$ and a pure-$B$ actor remain.  Since
$A\cap B=\emptyset$, the partition $\{A\}\sqcup\{B\}$ is the second case of
the lemma, and Corollary~\ref{cor:realizedmediator} gives the same separation,
every realized $A$--$B$ path passing through a mixed actor.  Internal cohesion
of each component is (iii).
\end{proof}

\begin{proof}[Proof of Proposition~\ref{prop:nonidentity}]
On that slice, by
Proposition~\ref{prop:supporthazard} the lifetimes are
$T_{AB}\sim\mathrm{Exp}(\lambda(\ell_A+\ell_B))$,
$T_A\sim\mathrm{Exp}(\lambda\ell_A)$ and
$T_B\sim\mathrm{Exp}(\lambda\ell_B)$.  Since $\ell_A,\ell_B\ge1$ the rate of
$T_{AB}$ strictly exceeds both others, so
$T_{AB}<_{\mathrm{st}}T_A$ and $T_{AB}<_{\mathrm{st}}T_B$: selective attrition holds.  At
$\beta=0$ the axis-common removal channel of Definition~\ref{def:law} has
intensity $\beta\lambda=0$ and is therefore inactive, so correlated collapse fails.  Two
predicates that agree on every regime cannot disagree on this one.
\end{proof}

\begin{proof}[Proof of Proposition~\ref{prop:enrichment}]
The displayed identity is Bayes' rule applied to the reference composition; the
denominator is $\Pr(T\le t)$ and is positive because $\ell_A,\ell_B\ge1$ and
$t>0$ make each $p_s(t)$ positive.  Subtracting,
\[
  p_{AB}-\bigl(\pi_A p_A+\pi_{AB}p_{AB}+\pi_B p_B\bigr)
  \;=\;\pi_A\bigl(p_{AB}-p_A\bigr)+\pi_B\bigl(p_{AB}-p_B\bigr),
\]
using $\pi_A+\pi_{AB}+\pi_B=1$.  Both differences are strictly positive by
Corollary~\ref{cor:finitehorizon}, and $\pi_A+\pi_B=1-\pi_{AB}>0$ makes at least
one weight strictly positive, so the right-hand side is strictly positive.
\end{proof}

\begin{proof}[Proof of Lemma~\ref{lem:persistmax}]
Fix $x>0$.  Let $M_{\mathrm{hi}}$ be the maximum lifetime among the actors with
signatures outside $\mathcal A$, and for $s\in\mathcal A$ let
$M^{(s)}_{\mathrm{lo}}$ be the maximum lifetime among the $n_s$ actors with signature $s$.
A union bound over the actors outside $\mathcal A$ gives
\[
\Pr(M_{\mathrm{hi}}>x)\;\le\;\sum_{s\notin\mathcal A}n_s\,e^{-\lambda L_sx},
\]
and independence within each class gives, using $1-u\le e^{-u}$,
\[
\Pr\bigl(M^{(s)}_{\mathrm{lo}}\le x\bigr)
=\bigl(1-e^{-\lambda L_sx}\bigr)^{n_s}
\;\le\;\exp\bigl(-n_s\,e^{-\lambda L_sx}\bigr).
\]
On the intersection of $\{M_{\mathrm{hi}}\le x\}$ with
$\{M^{(s)}_{\mathrm{lo}}>x\}$ over $s\in\mathcal A$ we have
$M_{\mathrm{hi}}<M^{(s)}_{\mathrm{lo}}$ for every $s\in\mathcal A$, which is
$\mathrm{LAST}(\mathcal A)$; a union bound over the complements gives
\eqref{eq:persistbound}.  At time $M_{\mathrm{hi}}$ every actor outside
$\mathcal A$ has exited, so the surviving signature set is contained in
$\mathcal A$, and every $s\in\mathcal A$ still has a surviving actor, so it
contains $\mathcal A$; the two are therefore equal.

For the limit, write
$G:=\min_{s\in\mathcal A}W_s-\max_{s\notin\mathcal A}W_s$ and suppose
$G\to\infty$.  If $L_s=0$ for some $s\in\mathcal A$ that class never exits and
its term is $0$, so assume every $L_s>0$ and choose
\[
  x:=\tfrac12\Bigl(\min_{s\in\mathcal A}W_s+\max_{s\notin\mathcal A}W_s\Bigr).
\]
For $s\notin\mathcal A$ we then have
$\lambda L_sx-\log n_s=\lambda L_s(x-W_s)\ge\lambda L_sG/2\to\infty$, so the
first sum tends to $0$; for $s\in\mathcal A$ we have
$n_se^{-\lambda L_sx}=e^{\lambda L_s(W_s-x)}\ge e^{\lambda L_sG/2}\to\infty$, so
the second sum tends to $0$.  Since $S$ is finite, both sums have finitely many
terms and \eqref{eq:persistbound} gives $\Pr[\mathrm{LAST}(\mathcal A)]\to1$.

If $\mathcal A$ separates and $s^*\in\operatorname*{argmax}_{s\in S}W_s$ lay
outside $\mathcal A$, then $\max_{s\notin\mathcal A}W_s\ge W_{s^*}\ge
\min_{s\in\mathcal A}W_s$ and $G\le0$, so no separating set omits a maximizer.

A singleton
$\mathcal A=\{s^*\}$ forces every survivor to have signature $s^*$.  For the second,
every survivor has a signature in $\mathcal A$, so every edge of $G_C$
present at that time joins two survivors; by
Lemma~\ref{lem:signatureintersection} no such edge crosses the partition, since
signatures on opposite sides are disjoint.  Both parts are nonempty among the
survivors, so $G_C$ has at least two components.
\end{proof}

\begin{proof}[Proof of Corollary~\ref{cor:loadmin}]
Since $S$ is finite and $L_s>L^*$ for every $s\notin\mathcal A$, we have
$\delta\ge1$ because each $L_s$ is a count of deficient axes; if
$S=\mathcal A$ the claim is immediate, so assume
$S\setminus\mathcal A\neq\emptyset$.  If $L^*=0$ then $W_s=\infty$ for every
$s\in\mathcal A$; otherwise (H2) gives $W_s=\log n_s/(\lambda L^*)\to\infty$
there.  For $s\notin\mathcal A$, (H1) gives
$W_s\le\log C/\bigl(\lambda(L^*+\delta)\bigr)$, a bound not growing with the
population.  Hence $\min_{s\in\mathcal A}W_s-\max_{s\notin\mathcal A}W_s\to\infty$, so
Lemma~\ref{lem:persistmax} applies.  Substituting $n_s\le C$ and
$L_s\ge L^*+\delta$ for $s\notin\mathcal A$, and $L_s=L^*$ for
$s\in\mathcal A$, into \eqref{eq:persistbound} gives the displayed bound.
\end{proof}

\begin{proof}[Proof of Proposition~\ref{prop:simultaneous}]
The next transition is the clock attaining the minimum among independent
exponentials, and that minimum is attained by a given clock with probability
proportional to its rate.  The $g$-common clock is present because
$D_g\neq\emptyset$, and its rate is $\beta\lambda>0$; the remaining clocks are
finite in number with finite rates, so $\Lambda(x)<\infty$.  Definition~\ref{def:law}
removes every actor in $D_g$ on that transition, and $|D_g|\ge2$.
\end{proof}

\begin{proof}[Proof of Proposition~\ref{prop:routeselector}]
On the frozen slice $\nu=\rho=0$, so $\varphi$ does not move and deficiency is
fixed axis by axis until an exit clock removes the deficient actors
themselves.  Axes belonging to the same requirement mask are required by
exactly the same actors (Section~\ref{sec:generation}), so every deficient
axis $g$ in mask $A$ is required by exactly the $A$- and $AB$-actors and is
deficient for all of them, giving $D_g=A\cup AB$; symmetrically $D_g=B\cup AB$
for every deficient axis in mask $B$.

By Definition~\ref{def:law} the idiosyncratic exit clocks fire at rate
$(1-\beta)\lambda$, which is $0$ at $\beta=1$, proving (1).  What remains are
the axis-common clocks, one per axis $g$ with $D_g\neq\emptyset$, each an
independent $\mathrm{Exp}(\beta\lambda)=\mathrm{Exp}(\lambda)$ clock that
removes every actor in $D_g$ at once.  There are $\ell_A$ such clocks
attached to mask $A$ and $\ell_B$ attached to mask $B$.

Let $\tau_A$ be the minimum of the $\ell_A$ mask-$A$ clocks and $\tau_B$ the
minimum of the $\ell_B$ mask-$B$ clocks.  As a minimum of independent
$\mathrm{Exp}(\lambda)$ clocks, $\tau_A\sim\mathrm{Exp}(\lambda\ell_A)$ and
$\tau_B\sim\mathrm{Exp}(\lambda\ell_B)$, independently, and
$\tau=\min\{\tau_A,\tau_B\}$.  These are continuous distributions, so
$\tau_A\neq\tau_B$ almost surely and a single axis attains $\tau$.  If that
axis lies in mask $A$, its clock removes $D_g=A\cup AB$ at once by
Definition~\ref{def:law}; every $A$- and $AB$-actor alive just before $\tau$
belongs to this set, since neither idiosyncratic exit nor any other clock has
touched it, so all of them are removed together and only $B$-actors remain
at $\tau+$.  The case where the firing axis lies in mask $B$ is symmetric.
This proves (2).

For (3), the same competing-clocks fact used in
Proposition~\ref{prop:simultaneous} gives, for two independent exponentials,
that the smaller one is the first with probability proportional to its rate:
\[
  \Pr(\tau_A<\tau_B)=\frac{\lambda\ell_A}{\lambda\ell_A+\lambda\ell_B}
  =\frac{\ell_A}{\ell_A+\ell_B} ,
\]
which is exactly $\Pr(A\text{-side first})$.
\end{proof}

\begin{proof}[Proof of Proposition~\ref{prop:divergentendpoints}]
For (i), apply Lemma~\ref{lem:persistmax} with $S=\{A,AB,B\}$ and
$\mathcal A=\{A,B\}$.  Here $L_A=L_B=\ell$ and $L_{AB}=2\ell$, and the
separation hypothesis for this $\mathcal A$ is the thinness assumed, so the last
exit outside $\mathcal A$ is $\tau$ and $A\cap B=\emptyset$ supplies the
partition of the second case.  For
(ii), Proposition~\ref{prop:routeselector} applies.  Since $\ell\ge1$ the mixed
class is deficient on axes of both masks, so the first axis-common firing
removes it whichever mask that axis lies in; that firing is $\tau$, and part (2) of
Proposition~\ref{prop:routeselector} leaves one class at $\tau+$, part (3) making the two outcomes equally likely
because $\ell_A=\ell_B$.  For (iii), the marginal law
$T_s\sim\mathrm{Exp}(\lambda L_s)$ of Proposition~\ref{prop:supporthazard} does
not depend on $\beta$.
\end{proof}

\subsection{Proofs from Section~\ref{sec:boundaries}}\label{app:proofs-boundaries}

\begin{proof}[Proof of Lemma~\ref{lem:unioncut}]
If $u\in S$ and its removal disconnects $s$ from $t$, then $u$ is a separator.
For any other $v\in S\setminus\{s,t,u\}$ we have $s\cap u=s\neq\emptyset$ and
$u\cap t=t\neq\emptyset$, so the path $s-u-t$ survives the removal of $v$ and
$v$ is not a separator.  Hence $u$ is the unique cut vertex.  Conversely, if $u$ is the
unique cut vertex then in particular $u\in S$ and removing it disconnects $s$
from $t$.
\end{proof}

\begin{proof}[Proof of Theorem~\ref{thm:failasymptotic}]
Fix a disjoint pair $\{s,t\}$ and write $q:=|\Sigma_r|=2^{r}-1$ and
$B:=B(s,t)$.  Since $u\cap s=s$ and $u\cap t=t$ are nonempty, $u$ lies in $B$,
and $u\notin\{s,t\}$ because both are nonempty and disjoint; neither $s$ nor $t$
lies in $B$, for the same disjointness.  Forced coincidence requires $u\in S$,
so $S\cap B$ contains $u$, while Lemma~\ref{lem:twobridges} forbids a second
element.  Hence $S\cap B=\{u\}$ exactly, and a realized family containing $s,t$ with
$\mathrm{GOOD}$ is free only on the $q-2-|B|$ signatures outside
$\{s,t\}\cup B$: there are at most $2^{\,q-2-|B|}$ of them.  For the
denominator, every realized family containing $[r]$ is connected, since $[r]$ meets
every nonempty signature; hence at least $2^{\,q-3}$ realized families containing $s,t$
are admissible.  The ratio is at most $2^{\,1-|B|}$, which is decreasing in
$|B|$, so Lemma~\ref{lem:bridgesize} gives the bound with $|B|$ replaced by
$2^{\,r-2}$.  The bound is uniform over pairs, and a weighted average of ratios
is at most the largest ratio, so it transfers to $\Omega_r$.
\end{proof}

\begin{proof}[Proof of Proposition~\ref{prop:forced-edge-density-cap}]
By Corollary~\ref{cor:kappaform} the vertex $u$ is a cut vertex of $G[S]$, so by
the block and cut-vertex structure of \citet{whitney1932} the graph
$G[S]\setminus\{u\}$ splits the remaining vertices into at least two
components; let $\mathcal S$ be the one containing $s$ and let $\mathcal T$ be
the union of the others, which contains $t$.  By
Remark~\ref{rem:permission} an edge of $G_C$ joins actors whose signatures
intersect, so $G_C$ is a subgraph of the graph obtained from $G[S]$ by replacing
each vertex $v$ with $n_v$ actors, each class by a clique and each edge by the
complete join.  No edge of that graph joins an actor whose signature lies in
$\mathcal S$ to one whose signature lies in $\mathcal T$.  Writing $a:=\sum_{v\in\mathcal S}n_v$ and
$b:=N-n_u-a$, at least $ab$ pairs are therefore non-adjacent in $G_C$.  The sum
$a+b=N-n_u$ is fixed and $a$ ranges over $[\,n_s,\;N-n_u-n_t\,]$, on which the
concave function $a\mapsto a\,b$ attains its minimum at an endpoint; the two
endpoint values are $n_s(N-n_u-n_s)$ and $n_t(N-n_u-n_t)$.  Dividing by
$\binom N2$ gives the bound.  At unit composition $N=m$ and
$n_u=n_s=n_t=1$, so the minimum is $m-2$ and the bound reads
$\sigma\le\bigl(\binom m2-(m-2)\bigr)/\binom m2=\mathrm{cap}(m)$.
\end{proof}

\begin{proof}[Proof of Proposition~\ref{prop:coverbounds}]
Let $D_C$ be the graph on $C$ joining pairs with disjoint signatures.  No $r+1$
nonempty subsets of $[r]$ are pairwise disjoint, so $\omega(D_C)\le r$ and $D_C$
is $K_{r+1}$-free; by Tur\'an's theorem $|E(D_C)|\le t_r(n)$.  Within a
component $Q(C)=\binom n2$, and the pairs counted by $A(C)$ are exactly the
non-edges of $D_C$, so
$A(C)=\binom n2-|E(D_C)|\ge\binom n2-t_r(n)$, which gives the bound; the
asymptotic form follows from $t_r(n)=(1-1/r)\binom n2+O(n)$.  Finally
$P_A(C)=1$ holds exactly when $D_C$ has no edge, that is, when the signatures in
$C$ form a pairwise intersecting family in the sense of
\citet{frankltokushige2018}.
\end{proof}

\subsection{Proofs from Section~\ref{sec:observables}}\label{app:proofs-observables}

\begin{proof}[Proof of Proposition~\ref{prop:balancewindow}]
Write $a=n_A^{(N)}$, $b=n_B^{(N)}$ and $S=S_N$.  With two components,
\[
  I_N=\frac{\binom a2+\binom b2}{\binom S2}
     =\frac{a^{2}+b^{2}-S}{S(S-1)} ,
\]
so $I_N<\tfrac12$ is equivalent to $2(a^{2}+b^{2}-S)<S(S-1)$, hence to
$a^{2}+b^{2}-2ab<S$, hence to $(a-b)^{2}<S$, hence to $|Z_N|<1$.  This
equivalence is exact at every $N$ with $S_N\ge2$, which is what
$S_N\to\infty$ in probability secures; $I_N$ is undefined below that.  Since $\Pr(|Z|=1)=0$, the set $\{|z|<1\}$
is a continuity set of the law of $Z$, so the portmanteau theorem gives
$\Pr(|Z_N|<1)\to\Pr(|Z|<1)$.
\end{proof}

\begin{proof}[Proof of Corollary~\ref{cor:modelbalance}]
At $\beta=0$ the lifetimes are independent, and by
Proposition~\ref{prop:supporthazard} the mixed actor has $L_{AB}=2\ell$, so it
exits at a time $\tau\sim\mathrm{Exp}(2\alpha)$, while each pure actor survives
past $\tau$ with probability $\theta:=e^{-\alpha\tau}$.  Conditionally on $\tau$
the two surviving pure counts $n_A^{(N)},n_B^{(N)}$ are independent
$\mathrm{Bin}(N,\theta)$, so their difference has mean zero and variance
$2N\theta(1-\theta)$ while $S_N/(2N\theta)\to1$ in probability; the
Lindeberg--Feller theorem and Slutsky's lemma give
$Z_N\mid\tau\Rightarrow N(0,1-\theta)$.  For the mixing law,
$\Pr(\theta\le x)=\Pr\bigl(\tau\ge-\alpha^{-1}\log x\bigr)=x^{2}$ on $(0,1)$,
which is $\mathrm{Beta}(2,1)$.  Mixing gives the
stated representation, and $Z$ has a density, so $\Pr(|Z|=1)=0$.  The displayed
expectation is $\Pr(|Z|<1)$ conditioned on $\theta$ and averaged.  Since
$0<\theta<1$ almost surely, $(1-\theta)^{-1/2}>1$ almost surely, which places the
expectation strictly above $2\Phi(1)-1$ and strictly below $1$; integrating
against the $\mathrm{Beta}(2,1)$ density $2\theta$ gives the stated decimal.
Finally, by
Theorem~\ref{thm:thinbridge} the event that both pure classes still have a
survivor when the mixed actor exits has probability tending to one, so
restricting to it does not change the limit.
\end{proof}

\input{sections/A_coordinates_proofs}

\section{Numerical audit trail}\label{sec:appendix}

The numerical statements in the body are recorded here, each with the domain it
was computed over.  Every domain is exhausted except where noted.

\subsection{Inverse observability: supports}\label{app:supports}

Domain: $r=4$, realized families of size $3\le|S|\le6$, unit composition, every nonzero
pattern of axis satisfaction, disconnected realized families admitted with the pair
required only to lie in a common component.  This gives $\pvInvGoodCount$ forced and
$\pvInvFailCount$ non-forced instances.  Entries are the minimum and maximum attained.

\begin{center}
\input{tables/support}
\end{center}

\noindent
Three features underlie the claims of Section~\ref{sec:observables}.  The maximum
of $\sigma$ under forcing equals $\mathrm{cap}(m)$ at every $m$; the non-forced
maximum exceeds it from $m=4$ on.  For $I$ the two supports share endpoints at
every $m$: no one-dimensional threshold was found.  For $K$ and $P_A$ the
supports coincide at $m=3$ and differ from $m=4$ on one or both sides, except
$P_A$ at $m=6$; no differing endpoint matches a closed form.

\subsection{Where the observables separate}\label{app:sweeps}

On the slice with every axis satisfied and one actor per signature, all
$\pvObservablesCollapsed$ connected realized families with $r\le4$ satisfy $\sigma=K=P_A$ and $I=1$:
the observables cannot be told apart there, so the sweeps below leave that
slice.

\begin{center}
\begin{tabular}{@{}p{0.38\linewidth}rrrr@{}}
\toprule
Sweep & forced instances & $\sigma$ & $K$ & $P_A$ \\
\midrule
axis satisfaction allowed to vary, unit composition & $\pvCapGoodHZero$ & 0 & 0 & $\pvCapViolPAHZero$ \\
composition $n_s\in\{1,\dots,4\}$ sampled, satisfaction varied & $\pvCapGoodComposition$ & $\pvCapViolSigmaComposition$ & $\pvCapViolSigmaComposition$ & --- \\
the same, against the composition-aware bound & $\pvCapGenInstances$ & $\pvCapGenViolSigma$ & --- & --- \\
\bottomrule
\end{tabular}
\end{center}

\noindent
Entries count violations of the cap named in the row.  The first two rows use
$\mathrm{cap}(m)$: the first shows it does not extend to $P_A$, the second that
it fails away from unit composition.  The third uses the composition-aware
bound of Proposition~\ref{prop:forced-edge-density-cap} on $\pvCapGenInstances$
sampled instances, with forcing read off $H_r[S]$, and finds no violation;
$\pvCapGenViolUnit$ of those instances exceed $\mathrm{cap}(m)$ and
$\pvCapGenTight$ attain the bound with equality.  On the $\pvCapGenUnitCases$
forced instances of that domain the two agree at unit composition without
exception.  The rows with varying composition are sampled, not exhaustively
enumerated.
Among pointed-graph isomorphism classes with $r\le4$ and $|S|\le7$,
$\pvNonInvariantIsoClasses$ contain both a forced and a non-forced instance:
the computation behind Remark~\ref{rem:notinvariant}.

\subsection{Sharpness of the edge-density cap}\label{app:extremal}

The unit-composition bound $\mathrm{cap}(m)$ of
Proposition~\ref{prop:forced-edge-density-cap} is attained for every $m$ once
the lattice is large enough.  Take $s=\{1\}$, $t=\{2\}$, $u=\{1,2\}$ and, for
the remaining $m-3$ signatures, distinct subsets of $\{2,\dots,r\}$ containing
$2$: each meets $u$ and misses $s$ and they meet one another, so the family
realizes $K_{m-1}$ with $s$ pendant at $u$.  This needs $m-2\le2^{\,r-2}$, that
is $r\ge2+\lceil\log_2(m-2)\rceil$.

\begin{center}
\input{tables/extremal}
\end{center}

\noindent
Forcing holds at $(s,t)$ in each case, with $u$ the unique cut vertex; omitted
values of $m$ behave the same way.

\subsection{Enumeration of \texorpdfstring{$\Omega_r$}{Omega r}}\label{app:orbit}

The proportions underlying Section~\ref{sec:boundaries} were computed twice: by
direct enumeration of labelled realized families, and by quotienting under the symmetric
group on mask labels, each representative weighted by $r!/|\mathrm{Stab}|$.  The
first column counts \emph{candidate sets}: subsets $S\subseteq\Sigma_r$ with
$|S|\ge2$ before the connectedness requirement of Definition~\ref{def:omega},
namely $2^{|\Sigma_r|}-1-|\Sigma_r|$.  Connectedness and disjointness are
applied when the pairs are counted, so the second column is the slice of
$\Omega_r$ at that $r$ and the first is not.  Relabelling is a bijection on
admissible pairs and preserves forcing, so the two agree exactly.

\begin{center}
\input{tables/orbit}
\end{center}

\noindent
The last column counts orbit representatives, an artifact of the second method.
Stratified by $|S|$, the forced proportion is greatest at $|S|=3$:
$\pvStratForcedRThreeSThree/\pvStratPairsRThreeSThree$ at $r=3$ and
$\pvStratForcedRFourSThree/\pvStratPairsRFourSThree$ at $r=4$, the latter above
the bound of Theorem~\ref{thm:failasymptotic}; at $|S|=4$ the values are
$\pvStratForcedRThreeSFour/\pvStratPairsRThreeSFour$ and
$\pvStratForcedRFourSFour/\pvStratPairsRFourSFour$.
The other three columns (candidate sets, admissible pairs, non-forced) agree
between the two methods at each $r$.  Direct enumeration does
not reach $r=5$, which is what Theorem~\ref{thm:failasymptotic} removes the need
for.

\section{Relation to neighbouring mathematical structures}\label{app:comparisons}

These are the comparisons announced in Section~\ref{sec:discussion}.  Each of
the three lines of work below supplies part of the characterization of
Section~\ref{sec:boundaries}.  The same thing is missing in each case: a
separator named in advance by the requirement lattice.

\begin{center}
\begin{tabular}{@{}p{0.24\linewidth}p{0.34\linewidth}p{0.34\linewidth}@{}}
\toprule
Prior work & What it supplies & Difference from the present setting \\
\midrule
Intersection graphs of lattice and ring ideals
 \citep{zelinka1973,chakrabarty2009ideals,sharmabhatwadekar1995,yewu2012comaximal}
 & a cut-vertex characterization, via Menger, for the order diagram recovered
   from the graph, and the same construction over rings
 & the condition is global and concerns the recovered diagram.  Where a join
   does appear it is the adjacency rule, not a vertex whose presence and
   separating role are in question \\
Quorum systems \citep{naorwool1998}
 & a load derived from an intersection axiom and bounded below, which is the
   move made in Section~\ref{sec:generation}
 & the load is a minimax over the whole universe.  No pair is designated and no
   separation appears \\
Vetoers and the Nakamura number \citep{nakamura1979,freixaskurz2019}
 & an indispensable position derived from an intersection structure, with a
   threshold in a dimension parameter as ours has in the number of masks
 & the vetoer is defined globally, and what concentrates on it is power rather
   than burden \\
\bottomrule
\end{tabular}
\end{center}

\noindent
Two further conditions on set families stop short of the same point.
Union-closed families \citep{bruhnschaudt2015} ask closure at every pair
rather than at a designated one.  Conformality \citep{zykov1974,boros2025conformal}
is indexed by triples and is existential over the family, and in a survey of the
Helly property \citep{dourado2009helly} we found no characterization in which a
designated member disconnects a designated pair.

\noindent
One neighbouring result makes a further point.  For a ring with no strongly
isolated maximal ideal, \citet[Lemma~3.7]{samei2014comaximal} shows that a
common neighbour of two vertices is always accompanied by a second.  A
nonadjacent pair with a common neighbour then has local vertex connectivity at
least two.  Pairwise algebraic structure forces redundancy where
Section~\ref{sec:boundaries} requires uniqueness.

%% file: sections/A_coordinates_proofs.tex
\subsection{Proofs from Section~\ref{sec:coordinates}}\label{app:proofs-coordinates}

\begin{proof}[Proof of Proposition~\ref{prop:coord-lawinvariance}]
Definition~\ref{def:law} gives a Markov jump process on a finite state space
with constant rates, so it suffices to show that its generator, applied to a
function $f$ of $(\varphi_i,\xi_i)$ alone, returns an expression in
$(\varphi_i,\xi_i)$ that contains neither $\rho_S$ nor $\beta$.  Five kinds of
clock can move those coordinates.

\emph{Shared loss at an axis $g$}, rate $\rho_S\nu$, enabled when
$P_g\neq\emptyset$.  It moves $\varphi_{ig}$ exactly when $i\in P_g$, that is
when $\xi_i=1$, $g\in s_i$ and $\varphi_{ig}=1$; and $i\in P_g$ already forces
$P_g\neq\emptyset$, so on the states where the term is nonzero the enabling
condition is a condition on $i$ alone.  The contribution is
$\rho_S\nu\,[f(\varphi_{ig}\!\to\!0)-f]$.

\emph{Idiosyncratic loss at $(i,g)$}, rate $(1-\rho_S)\nu$, enabled on the same
set, contributing $(1-\rho_S)\nu\,[f(\varphi_{ig}\!\to\!0)-f]$.  The two sum to
$\nu\,[f(\varphi_{ig}\!\to\!0)-f]$.

\emph{Axis-common exit at $g$}, rate $\beta\lambda$, enabled when
$D_g\neq\emptyset$.  It removes $i$ exactly when $i\in D_g$, that is when
$\xi_i=1$, $g\in s_i$ and $\varphi_{ig}=0$; again $i\in D_g$ forces
$D_g\neq\emptyset$.  The contribution is $\beta\lambda\,[f(\xi_i\!\to\!0)-f]$
per deficient axis of $i$.

\emph{Idiosyncratic exit at a deficient incidence $(i,g)$}, rate
$(1-\beta)\lambda$, contributing $(1-\beta)\lambda\,[f(\xi_i\!\to\!0)-f]$ per
deficient axis.  The two exit terms sum to
$\lambda\,d_i\,[f(\xi_i\!\to\!0)-f]$.

The supply clock at $(i,g)$ with $i\in D_g$ has rate $\rho$ and contributes
$\rho\,[f(\varphi_{ig}\!\to\!1)-f]$.  Every remaining clock either fires on
indices disjoint from $i$ or, when shared, leaves $(\varphi_i,\xi_i)$ fixed
because $i\notin P_g$ or $i\notin D_g$, and contributes $0$.  Hence
\begin{multline*}
  (\mathcal{A}f)(\varphi_i,\xi_i)
  =\xi_i\Bigl[\nu\!\!\sum_{g\in s_i:\varphi_{ig}=1}\!\!\bigl[f(\varphi_{ig}\!\to\!0)-f\bigr]\\
   +\rho\!\!\sum_{g\in s_i:\varphi_{ig}=0}\!\!\bigl[f(\varphi_{ig}\!\to\!1)-f\bigr]
   +\lambda d_i\bigl[f(\xi_i\!\to\!0)-f\bigr]\Bigr].
\end{multline*}
With a deterministic initial state the law of
$(\varphi_i,\xi_i)_{t\ge0}$ is determined by this generator, and is therefore the
same for all $(\rho_S,\beta)$.
\end{proof}

\begin{proof}[Proof of Proposition~\ref{prop:coord-lossdependence}]
Since $\rho=0$, every incidence moves at most once, from $1$ to $0$.  Draw
$F_g\sim\mathrm{Exp}(\rho_S\nu)$ for each axis, $X_{ig}\sim
\mathrm{Exp}((1-\rho_S)\nu)$ for each incidence with $\varphi_{ig}(0)=1$, and
$E_i\sim\mathrm{Exp}(1)$ for each actor, all independent, with
$\mathrm{Exp}(0):=\infty$.  Put $L_{ig}:=\min(F_g,X_{ig})\sim\mathrm{Exp}(\nu)$.
Before $T_i$ the actor $i$ survives, so $\varphi_{ig}=1$ puts $i\in P_g$ and
keeps the shared clock of $g$ active; $H_i$ is evaluated only on $[0,T_i)$, so by the
lack of memory the first firing of $F_g$ while $i\in P_g$ is the loss time of
$(i,g)$.  With $\beta=0$ the exit clocks are
the idiosyncratic ones, so the cumulative exit intensity of $i$ is
\[
  H_i(t)=\lambda\Bigl[d_i(0)\,t+\sum_{g\in A_i}(t-L_{ig})_+\Bigr],
  \qquad T_i=\inf\{t:H_i(t)\ge E_i\}.
\]
Write $W:=(F_g)_{g\in G}$.  Conditionally on $W$ the families
$\{(X_{ig})_g,E_i\}$ are independent across actors and $T_i$ is a function of
$W$ and of the family of $i$, so $T_i\perp T_j$ given $W$.  If $\rho_S=0$ then
$W$ is degenerate and the conditioning is vacuous; if $C=\emptyset$ then $T_i$
depends on $W$ only through $(F_g)_{g\in A_i}$ and $T_j$ only through
$(F_g)_{g\in A_j}$, two independent families.  Either way $T_i\perp T_j$.

Now let $C\neq\emptyset$ and $\rho_S>0$, fix $s,t>0$ and put
$u(W):=\Pr(T_i>s\mid W)$ and $v(W):=\Pr(T_j>t\mid W)$.  Conditional
independence gives $D(s,t)=\mathrm{Cov}\bigl(u(W),v(W)\bigr)$.  Raising $F_g$
can only delay a loss, hence only decrease $H_i$ pointwise, so $u$ and $v$ are
nondecreasing in each coordinate of $W$; $u$ does not depend on $F_g$ for
$g\notin A_i$, and likewise for $v$.  Fix $g_0\in C$ and write $W_{-}$ for the
remaining coordinates.  Then
\[
  \mathrm{Cov}(u,v)=\mathbb{E}\bigl[\mathrm{Cov}(u,v\mid W_{-})\bigr]
   +\mathrm{Cov}\bigl(\mathbb{E}[u\mid W_{-}],\mathbb{E}[v\mid W_{-}]\bigr),
\]
and the second term is nonnegative by induction on the number of coordinates,
each step being Chebyshev's association inequality.  For the first term,
condition further on $(X_{ig})_g$ and $(X_{jg})_g$: on
$\{X_{ig_0}>x\}$ the map $x\mapsto(s-\min(x,X_{ig_0}))_+$ is strictly
decreasing for $x<s$, so $u$ is strictly increasing in $F_{g_0}$ on $(0,s)$, and
$v$ is strictly increasing on $(0,t)$.  Since $\rho_S\nu>0$, $F_{g_0}$ charges
$(0,s\wedge t)$, so Chebyshev's inequality is strict there and
$\mathrm{Cov}(u,v\mid W_{-})>0$ with positive probability.  Hence $D(s,t)>0$.
\end{proof}

\begin{proof}[Proof of Proposition~\ref{prop:coord-supply}]
(i) With $\nu=0$ a satisfied incidence stays satisfied, so each of the $L$
incidences deficient at time zero carries a supply clock
$R_g\sim\mathrm{Exp}(\rho)$ and, until $R_g$ fires, an exit clock of rate
$\lambda$; all clocks are independent.  Given $(R_g)_g$ the exit clock at $g$ is
active on $[0,R_g)$, so
$\Pr(T>t\mid(R_g)_g)=\exp(-\lambda\sum_g(t\wedge R_g))$, and
\[
  \mathbb{E}\bigl[e^{-\lambda(t\wedge R)}\bigr]
  =e^{-(\lambda+\rho)t}+\int_0^t\!\rho e^{-\rho x}e^{-\lambda x}\,dx
  =\frac{\rho+\lambda e^{-(\lambda+\rho)t}}{\lambda+\rho}=h(t).
\]
Taking the product over the $L$ incidences gives $\Pr(T>t)=h(t)^L$, and letting
$t\to\infty$ gives $q^L$.  Since $h(t)\in(q,1)$ for $t>0$, both are strictly
decreasing in $L$.  The clock index sets are per-incidence and disjoint across
actors, and the two clocks with shared index sets are inactive, the shared
loss at rate $\rho_S\nu=0$ and the axis-common exit at rate $\beta\lambda=0$;
hence the lifetimes are independent across actors.

(ii) An actor with signature $s$ exits with probability $1-q^{L_s}$, independently, which
gives the extinction probability and, by the law of large numbers applied to the
$\mathrm{Bin}(n_s,q^{L_s})$ counts of permanent survivors, the composition.  For
the last claim, $\mathrm{LAST}(\mathcal A)$ implies that every actor outside
$\mathcal A$ exits, and is implied by that event together with the event that
each $s\in\mathcal A$ has a permanent survivor; by independence
\[
  \prod_{s\notin\mathcal A}(1-q^{L_s})^{n_s}
  \prod_{s\in\mathcal A}\bigl(1-(1-q^{L_s})^{n_s}\bigr)
  \le\Pr[\mathrm{LAST}(\mathcal A)]
  \le\prod_{s\notin\mathcal A}(1-q^{L_s})^{n_s},
\]
and the second factor on the left tends to $1$.

(iii) For one incidence, $\mathbb{E}[e^{-\lambda(t\wedge R)}\mathbf1\{R>t\}]
=e^{-(\lambda+\rho)t}$, so given survival to $\tau$ the incidences are
independently still deficient with the class-free probability
$\gamma(\tau):=e^{-(\lambda+\rho)\tau}/h(\tau)$, and by the lack of memory
\[
  p^{(\tau)}_s(t):=\Pr(T\le\tau+t\mid T>\tau)
  =1-\bigl(1-\gamma(\tau)(1-h(t))\bigr)^{L_s}.
\]
This is strictly increasing in $L_s$, so $p^{(\tau)}_{AB}>p^{(\tau)}_A,
p^{(\tau)}_B$ and the computation in the proof of
Proposition~\ref{prop:enrichment} applies unchanged, giving $E_{AB}>1$.  As
$\tau\to\infty$, $\gamma(\tau)\to0$ and $p^{(\tau)}_s=L_s\gamma(1-h(t))(1+o(1))$,
so $E_{AB}\to L_{AB}/\sum_s\pi_sL_s$ at every horizon; the at-risk composition
tends to the permanent-survivor composition of (ii), which at $\ell_A=\ell_B=1$
and equal multiplicities is proportional to $(q,q^2,q)$, and substituting
$L_A=L_B=1$, $L_{AB}=2$ gives $(2+q)/(1+q)$.  For the monotonicity write
$u:=2-\gamma(\tau)(1-h(t))\in(1,2)$ and $h:=h(\tau)$; the same substitution gives
$E_{AB}=u(2+h)/(2+hu)$, with $\partial_uE_{AB}=2(2+h)/(2+hu)^2>0$ and
$\partial_hE_{AB}=2u(1-u)/(2+hu)^2<0$, and $\gamma$ and $h$ both decrease in
$\tau$.

(iv) By (i) the three events are independent and an actor with load $L$ is present
at $t$ with probability $h(t)^L$.  Since $n_{AB}\ge1$ and $h\ge q$, the first
factor is at most $1-h(t)^{2\ell}\le1-q^{2\ell}$.

(v) With $\nu=0$ each $\varphi_{ig}$ is nondecreasing, and $R$ is nonincreasing.
Let $t_1\le t_2$, $i,j\in R(t_2)$ and $\{i,j\}\in E_C(t_1)$.  The axis
$g\in s_i\cap s_j$ that both satisfied at $t_1$ is satisfied by both at $t_2$,
so $\{i,j\}\in E_C(t_2)$ by Definition~\ref{def:cohesion}; an edge is lost
only through the exit of an endpoint.  Since $R_\infty:=\{i:T_i=\infty\}\subseteq R(t)$ for every $t$, the
inclusion applies at every pair of times, and $E_C$ restricted to $R_\infty$ is
nondecreasing.
On $\{T_i=\infty\}$ no exit
clock of $i$ ever fires, so every incidence of $i$ deficient at time zero is
supplied, and eventually $\varphi_{ig}=1$ for all $g\in s_i$.  Hence two
permanent survivors with $s_i\cap s_j\neq\emptyset$ are eventually adjacent.
\end{proof}

\begin{proof}[Proof of Proposition~\ref{prop:coord-loss}]
(i) With $\rho=0$ and $\beta=0$ take the feedback-free construction: draw
$\sigma_g\sim\mathrm{Exp}(\rho_S\nu)$ per axis,
$\varepsilon_{ig}\sim\mathrm{Exp}((1-\rho_S)\nu)$ per initially satisfied
incidence and $E_i\sim\mathrm{Exp}(1)$ per actor, all independent, and set
$\varphi_{ig}(t)=\mathbf1\{t<Z_{ig}\}$ with
$Z_{ig}:=\min(\sigma_g,\varepsilon_{ig})\sim\mathrm{Exp}(\nu)$ for $g\in A_i$,
and $T_i=\inf\{t:\lambda\int_0^td_i(u)\,du\ge E_i\}$.  The construction lets
losses continue after an exit, but Definition~\ref{def:cohesion} refers only to
$R(t)$, and exit times are unaffected.  The $Z_{ig}$, $g\in A_i$, are
independent, so
\[
  \Pr(T_i>t)=e^{-\lambda L_it}\prod_{g\in A_i}\mathbb{E}\bigl[e^{-\lambda(t-Z)_+}\bigr],
\]
where $Z\sim\mathrm{Exp}(\nu)$ and
\[
  \mathbb{E}\bigl[e^{-\lambda(t-Z)_+}\bigr]
  =e^{-\nu t}+\int_0^t\!\nu e^{-\nu z}e^{-\lambda(t-z)}\,dz=B(t),
\]
with the stated value at $\nu=\lambda$.  The marginal does not depend on $\rho_S$.  As $t\to\infty$,
$B(t)\sim\lambda e^{-\nu t}/(\lambda-\nu)$ for $\nu<\lambda$ and
$B(t)\sim\nu e^{-\lambda t}/(\nu-\lambda)$ for $\nu>\lambda$, so
$-t^{-1}\log\Pr(T_i>t)\to\lambda L_i+(k_i-L_i)\widehat\nu=\lambda
L^{\mathrm{eff}}_i$.

(ii) Couple $i$ and $j$ on one family: the two count inequalities give an
injection of the axes of $i$ into those of $j$ sending deficient to deficient
and satisfied to satisfied, and $i$ uses the corresponding $Z$ and the same $E$.
Then $d_j(u)\ge d_i(u)$ for every $u$ pathwise, so $T_j\le T_i$ pathwise.  For
the last sentence, let $i$ require the single axis $a$ and be deficient on it,
and let $j$ require $\{a,b\}$ and be satisfied on both, so $s_i\subseteq s_j$.
Then $\Pr(T_i>t)=e^{-\lambda t}$ and $\Pr(T_j>t)=B(t)^2$ with $B(0)=1$ and
$B'(0)=0$, so $\Pr(T_j>t)>\Pr(T_i>t)$ for all small $t>0$ and
$T_j\le_{\mathrm{st}}T_i$ fails.

(iii) At $\rho_S=0$ the lifetimes are independent with the survival function of
(i).  For $\nu\neq\lambda$ the identity
$B(t)=e^{-\widehat\nu t}\bigl[\max(\lambda,\nu)-\min(\lambda,\nu)
e^{-|\lambda-\nu|t}\bigr]/|\lambda-\nu|$ gives
$e^{-\widehat\nu t}\le B(t)\le\bigl(\max(\lambda,\nu)/|\lambda-\nu|\bigr)
e^{-\widehat\nu t}$, hence $e^{-\lambda L^{\mathrm{eff}}_st}\le\Pr(T>t)\le
C_se^{-\lambda L^{\mathrm{eff}}_st}$ with $C_s$ free of $n_s$, so the maximum of $n_s$ copies is $\log
n_s/(\lambda L^{\mathrm{eff}}_s)+O_P(1)$ by the usual extreme-value sandwich;
at $\nu=\lambda$ the polynomial factor costs a further $O_P(\log\log n_s)$.  At
$\rho_S>0$ put $\Sigma:=\max_g\sigma_g$, a maximum over the axes required by the
class and hence free of $n_s$.  The exit intensity never exceeds $\lambda k_s$,
so $T_i\ge E_i/(\lambda k_s)$; and for $t\ge\Sigma$ every axis of every actor
is deficient, so the intensity equals $\lambda k_s$ and $T_i\le\Sigma+
E_i/(\lambda k_s)$.  Taking maxima and using $\max_iE_i=\log n_s+O_P(1)$ gives
the claim.

(iv) For $s\notin\mathcal A$ use $\Pr(T>x)\le C_se^{-\lambda
L^{\mathrm{eff}}_sx}$ with
$C_s=\bigl(\max(\lambda,\nu)/|\lambda-\nu|\bigr)^{k_s-L_s}$, valid for all $x$
when $\nu\neq\lambda$, in the union bound (at $\nu=\lambda$ no such $C_s$
exists, the tail carrying the factor $(1+\lambda x)^{k_s-L_s}$, which is why
the part assumes $\nu\neq\lambda$); for $s\in\mathcal A$ use
$\Pr(T>x)\ge e^{-\lambda k_sx}$, from the intensity bound of (iii), together
with $1-u\le e^{-u}$.  The rest of the proof of Lemma~\ref{lem:persistmax} is
unchanged, and the separation condition is the divergence that sends both sums
to zero at $x$ chosen midway between the two scales.

(v) Each $Z_{ig}=\min(\sigma_g,\varepsilon_{ig})$ is $\mathrm{Exp}(\nu)$ with
$\nu>0$, hence finite almost surely, and the index set is finite, so
$\Theta'=\max_{(i,g)}Z_{ig}<\infty$ almost surely.  For $t\ge\Theta'$ every
initially satisfied incidence has been lost, and with $\rho=0$ the initially
deficient ones never recover, so $\varphi_{ig}(t)=0$ for every incidence and
Definition~\ref{def:cohesion} gives $E_C(t)=\emptyset$.

(vi) If $\rho_S>0$ then $Z_{ig}\le\sigma_g$, so $\Theta'\le\max_g\sigma_g$, whose
law does not depend on the multiplicities, while by (iii) the last mixed exit is
$\log n_{AB}/(\lambda k_{AB})+O_P(1)\to\infty$.  Hence the probability that the
bridge is exhausted before $\Theta'$ tends to zero, and on the complement every
surviving actor is isolated in $G_C$; with at least two pure-$A$ survivors at
that time the internal cohesion asserted by Theorem~\ref{thm:thinbridge} fails.
\end{proof}

%% file: tables/support.tex
\begin{tabular}{@{}llrrrr@{}}
\toprule
Observable & $m$ & \multicolumn{2}{c}{forced} & \multicolumn{2}{c}{non-forced} \\
\cmidrule(lr){3-4}\cmidrule(lr){5-6}
& & min & max & min & max \\
\midrule
$\sigma$ & 3 & 0.6667 & 0.6667 & 0.6667 & 0.6667 \\
         & 4 & 0.3333 & 0.6667 & 0.3333 & 0.8333 \\
         & 5 & 0.2000 & 0.7000 & 0.2000 & 0.9000 \\
         & 6 & 0.1333 & 0.7333 & 0.1333 & 0.9333 \\
\midrule
$I$      & 3 & 1.0000 & 1.0000 & 1.0000 & 1.0000 \\
         & 4 & 0.5000 & 1.0000 & 0.5000 & 1.0000 \\
         & 5 & 0.3000 & 1.0000 & 0.3000 & 1.0000 \\
         & 6 & 0.2000 & 1.0000 & 0.2000 & 1.0000 \\
\midrule
$K$      & 3 & 0.6667 & 0.6667 & 0.6667 & 0.6667 \\
         & 4 & 0.6667 & 0.6667 & 0.5000 & 0.8333 \\
         & 5 & 0.5000 & 0.7500 & 0.4000 & 0.9000 \\
         & 6 & 0.4000 & 0.7500 & 0.4000 & 0.9333 \\
\midrule
$P_A$    & 3 & 0.6667 & 0.6667 & 0.6667 & 0.6667 \\
         & 4 & 0.6667 & 0.8333 & 0.5000 & 0.8333 \\
         & 5 & 0.5000 & 0.9000 & 0.4000 & 0.9000 \\
         & 6 & 0.4000 & 0.9333 & 0.4000 & 0.9333 \\
\bottomrule
\end{tabular}

%% file: tables/extremal.tex
\begin{tabular}{@{}rrrrr@{}}
\toprule
$m$ & least $r$ & edges & $\binom{m-1}{2}+1$ & $\sigma=\mathrm{cap}(m)$ \\
\midrule
3 & 2 & 2 & 2 & 0.6667 \\
5 & 4 & 7 & 7 & 0.7000 \\
7 & 5 & 16 & 16 & 0.7619 \\
9 & 5 & 29 & 29 & 0.8056 \\
12 & 6 & 56 & 56 & 0.8485 \\
\bottomrule
\end{tabular}

%% file: tables/orbit.tex
\begin{tabular}{@{}rrrrr@{}}
\toprule
$r$ & candidate sets & admissible pairs & non-forced & orbits \\
\midrule
2 & $4$ & $1$ & $0$ & $3$ \\
3 & $120$ & $159$ & $132$ & $36$ \\
4 & $32{,}752$ & $202{,}198$ & $200{,}400$ & $1{,}987$ \\
\bottomrule
\end{tabular}